\documentclass[3p]{elsarticle}

\usepackage{hyperref}

\journal{arXiv}

\usepackage{amsfonts}
\usepackage{amsmath, amscd, amssymb, bm, amsbsy, epsf, amsthm}
    \usepackage{enumerate}
    \usepackage{paralist}
\usepackage{bbm, dsfont}
\usepackage{graphicx,color}
\usepackage{subfigure}
\usepackage{mathrsfs}
\usepackage{verbatim}
\usepackage{inputenc}
\usepackage{diagbox}
\usepackage{algorithm2e}
\usepackage{bbm, dsfont}
\usepackage{multicol}
\usepackage{balance}
\usepackage{verbatim, booktabs}

\usepackage{multirow, url}
\usepackage[table]{xcolor}

\usepackage{cmbright}
\usepackage{pifont}
\usepackage{inputenc}

\DeclareMathAlphabet{\mathpzc}{OT1}{pzc}{m}{it}
\DeclareMathAlphabet{\mathpzc}{OT1}{pzc}{m}{it}
\DeclareMathOperator*{\Var}{\mathbb{V}\!ar}
\DeclareMathOperator*{\unif}{Unif}
\DeclareMathOperator{\trace}{trace}

\newtheorem{theorem}{Theorem}
\newtheorem{lemma}[theorem]{Lemma}

\newtheorem{proposition}[theorem]{Proposition}
\newdefinition{definition}{Definition}
\newdefinition{hypothesis}{Hypothesis}
\newdefinition{problem}{Problem}
\newdefinition{remark}{Remark}
\newdefinition{example}{Example}
\newdefinition{analysis}{Analysis}

\newcommand{\eversor}{\mathbf{\widehat{e}}}
\newcommand{\msum}{\textup{sum}}

\newcommand\itS{\mathcal{S}}
\newcommand\itG{\mathcal{G}}
\newcommand\p{\mathbf{p}}
\newcommand\g{\mathbf{g}}

\newcommand\X{ \mathrm{X} } 
\newcommand\Y{ \mathrm{Y} } 
\newcommand\Z{ \mathrm{Z} } 
\newcommand\upns{ \mathrm{NS} } 
\newcommand\upne{ \mathrm{NE} } 
\newcommand\upnt{ \mathrm{NT} } 
\newcommand\boldsf{ \mathbf{sf} } 
\newcommand\itns{ \text{ns} } 

\newcommand\Exp{\mathbbm{E}}
\newcommand\prob{\boldsymbol{\mathbbm{P}}}

\newcommand\R{\boldsymbol{\mathbbm{R} } }
\newcommand\N{\boldsymbol{\mathbbm{N} } }
\newcommand\setZ{\boldsymbol{\mathbbm{Z}  } }
\newcommand\defining{\overset{\textbf{def}}=}

\newcommand\gpa{\mathop{}\!\mathit{GPA} } 
\newcommand\apv{\mathop{}\!\mathit{APV} }
\newcommand\ia{\mathop{}\!\mathit{IA} }
\newcommand\sa{\mathop{}\!\mathit{SA} }
\newcommand\mt{\mathop{}\!\mathit{mt} }

\begin{document}

\begin{frontmatter}

\title{A big data based method for pass rates optimization \\
in mathematics university lower division courses}
\tnotetext[mytitlenote]{This material is based upon work supported by project HERMES 45713 from Universidad Nacional de Colombia,
Sede Medell\'in.}

\author[mymainaddress]{Fernando A Morales } 
\cortext[mycorrespondingauthor]{Corresponding Author}
\ead{famoralesj@unal.edu.co}

\author[mysecondaryaddress]{Cristian C Chica}
\author[mymainaddress]{Carlos A Osorio}
\author[mymainaddress]{Daniel Cabarcas J}



\address[mymainaddress]{Escuela de Matem\'aticas
Universidad Nacional de Colombia, Sede Medell\'in \\
Carrera 65 \# 59A--110, Bloque 43, of 106,
Medell\'in - Colombia\corref{mycorrespondingauthor}}

\address[mysecondaryaddress]{Departamento de Econom\'ia, Universidad EAFIT. \\
Carrera 49 \# 7 Sur-50, Bloque 38, of 501, Medell\'in - Colombia}


\begin{abstract}
In this paper a big data based method is presented, for the enhancement of pass rates in mathematical university lower division courses, with several sections. We propose the student-lecturer partnership as the cornerstone of our optimization process. First, the students-lecturer success probabilities are computed using statistical segmentation and profiling techniques in the available historical data. Next, using integer programming models, the method finds the optimal pairing of students and lecturers, in order to maximize the success chances of the students' body. Finally, the analysis of our method as an economic process, as well as its importance for public universities of the third World, will be presented throughout the paper.
\end{abstract}


\end{frontmatter}



%
%

%
%
%
%
%
\section{Introduction}

It is well-known that improving pass rates in mathematics courses is of paramount importance for academic institutions all over the world. This objective becomes even more critical for public universities as they subside partially or totally the education of its enrolled undergraduates; subject to the country's legislation and the student's economic stratification.

The general consensus is that creating better conditions for the students will improve students' success. Hence, most of the work done in order to address this challenge has two principal directions: 
\begin{enumerate}[(i)]
\item The traditional pedagogical approach which, essentially aims to improve the presentation of the course contents on two fronts: presentation of mathematical discourse i.e., curricula reform, development of course materials and improvement of the lecturer's instructional practice. Part of the latter are the teaching evaluations' open questions, giving feedback to the instructor about how people felt during his/her classes. 

\item The uses of technology in the learning of mathematics. One stream goes in the recollection of data and the measurement of the resource impact in the cognitive process: development of LMS platforms and real-time feedback interfaces \cite{Theroux}. The other stream explores the use of the aforementioned harvested information to improve the learning process \cite{CastroEtAl}: targeted problem sets and training tests \cite{BeckMostow}, identification of favorable pedagogical approaches and learning patterns/styles \cite{LevyWilenski}, identification of fortitudes and weaknesses, assessment of study exercise vs skills building \cite{LevyWilenski}, problem solving approaches, platforms for interaction between students through the learning process \cite{HerrenkohlTasker}. There is also the use of big data to asses learning rather than improve instructional techniques, such as early detection of students at risk \cite{MaacfaydenDawson}. 
\end{enumerate}
The present work fits in the second category, in this case, the use of big data to define policies enhancing the higher education production \cite{Lazear} and without raising the costs. More specifically, the method will suggest an optimal design of student body/composition to maximize the pass rate chances. The design is driven by favorable teacher-student partnership, rather than peer diversity or a peer interaction criterion (see  \cite{DeGiorgiPellizaariWoolstonGui, DeGiorgiPellizaariRedaelli}). A second aspect of the method is that, it is based on the computation of expectations (conditioned to the students' segmentation) and not on statistical regression (linear or not), as it pursues the construction of a probabilistic model, rather than the construction of a production function (see  \cite{DeGiorgiPellizaariWoolstonGui, Kruegger}). In addition, the input database (which is the method's input), is updated from one academic term to the next, therefore, it seems more adequate to recompute the conditional expectations term-wisely instead of pursuing rigid regression models. Moreover, given the current computational tools and possibilities, once the method is coded as an algorithm, the proposed updating approach will come at zero cost increase. This is the main reason while our method will be presented and explained, mostly in the format of an algorithm.

Our approach starts from data bases of a public university for 15 semesters, between February 2010 and July 2017, containing the information of academic performance and demographics of its population. A first stage of descriptive statistics allows to identify the success factors by correlation. A second stage revisits the historical performance (15 semesters in total), it defines segmentation profiles and then computes the historical efficiency for each of the involved instructors, conditioned to the quantiles of segmentation. Next, it uses integer programming to find the optimal matching of students-instructors in two different ways. A third stage of the method randomizes the involved factors, namely the profile of the group taking the class, the number of Tenured Track instructors, the number of sections (group) in the course and such. This is aimed to produce Monte Carlo simulations and find the expected values of improvement in the long run. The factors are regarded as random variables with probabilistic distributions computed from the empirical knowledge, recorded in the database and each Monte Carlo simulation arises from one random realization of the algorithm.  
%
%
\subsection{Economic Justification}\label{economic_justification}
%
Whether or not public higher education constitutes a public good, is a subject that has been extensively debated in economics (see, e.g, \cite{grace1989education} and \cite{tilak2008higher}). On one side, it benefits the whole society by playing a redistributive role, where low income class students can access higher education to improve their future labor perspectives. On the other side, public higher education can be of limited access.\footnote{In Colombia, by 2016 only 49,42\% of the students had access to higher education (official information at \url{https://www.mineducacion.gov.co/portal/}.)  }. This last feature shows why the classification of public higher education as a public good, is a matter of debate among economists, because a public good must be accessible to all individuals. Such a debate is not the matter of this paper, but it highlights two important properties of public higher education (its redistribution role, and its limited access), which are relevant to understand the contribution of this work. 

The Colombian government has implemented different strategies to address the limited access aspect of education, implementing programs like ``Familias en Acci\'on'', a welfare program designed to improve school attendance rates among the young; and ``Ser Pilo Paga'', a merit-based financial aid program designed to increase the higher education attendance rates among the poorest. Although these strategies have increased the attendance rates in education (as showed by \cite{londono2017intended} in the case of higher education), a report from the World Bank (\cite{marta2017crossroads}) showed that 37 percent of students starting a bachelors' degree program withdraw from the higher education system (with percentages going as high as 53 percent, when including short-cycle programs). These numbers show that the limited access aspect of public higher education, can not be addressed only by means of increasing the raw coverage. Higher education institutions need strategies to decrease drop out rates, and require mechanisms to help students to improve their pass rates, grades and others. 

This paper provides a mechanism--an algorithm- to potentially help higher education institutions to improve such indicators of student welfare (pass rates and grades). 
Our approach is to understand the University, not only as an agent that provides education, but also as a \textbf{rational regulator agent}, capable to optimally allocate some of its resources for enhancement of social welfare of its students body.

Throughout the paper we will make remarks emphasizing the interpretation and/or importance of the problem from the economic point of view; these remarks will have a pertinent warning for the reader.
%
%
%
\subsection{Organization and notation}
%
%
The paper is organized as follows. In \textsc{Section} \ref{Sec The study case and its databases}, we a brief description of the study case setting and its databases. In \textsc{Section} \ref{Sec Quantification: Variables, Segmentation and Profiling} the available historical information is quantified by a process of statistical segmentation and profiling. In \textsc{Section} \ref{Sec Optimization and Historical Assessment} the two optimization mechanisms are presented, formulated a problems of integer programming and an assessment of the historical behavior is performed, i.e., compute the theoretical outcome should the method would have been applied in the past. In \textsc{Section} \ref{Sec Randomization and Asymptotic Assesment} we randomize the study case, using the historical records in order to generate random and plausible instances of the study case and apply the optimization method in order to perform Monte Carlo simulations, as well as observing its asymptotic behavior. Finally, \textsc{Section} \ref{Sec Conclusions} delivers the conclusions.

We close the introduction describing the mathematical notation. For any natural number $ N\in \N $, the symbol $ [N] \defining \{ 1, 2, \ldots , N \} $ indicates the set/window of the first $ N $ natural numbers. For any set $ E $ we denote by $ \vert E \vert $ its cardinal and $ \wp(E) $ its power set. We understand $ \Omega $ as a generic finite probability space $ \big( \Omega, \wp(\Omega), \prob \big) $  in which all outcomes are equally likely, i.e. the event probability function satisfies $ \prob(\{\omega\}) = \vert \Omega \vert^{-1} $ for all $ \omega \in \Omega $. In particular for any event $ E \subseteq \Omega $ it holds that 
\begin{equation}\label{Eq Probabiity Measure}
\prob( E ) \defining \dfrac{\vert E \vert}{\vert \Omega\vert } = \dfrac{\text{number of favourable outcomes}}{\text{total number of possible outcomes} } .
\end{equation}
%
%
A particularly important probability space is $ \itS_{N} $, where $ \itS_{N} $ denotes the set of all permutations in $ [N] $, its elements will be usually denoted by $ \pi, \sigma, \tau $, etc. Random variables will be represented with upright capital letters, namely $ \X, \Y, \Z, ... $, expectation and variance of such variables with $ \Exp(\X) $ and $ \Var(\X) $ respectively.
Vectors (deterministic or random) are indicated with bold letters, namely $ \p, \g,\mathbf{X} ... $ etc. Deterministic matrices are represented with capital letters  i.e., $ A, G, T $.
%
%
%
%
\section{The study case and its databases}\label{Sec The study case and its databases}
%
%
%
%
In this work our study case is the performance of lower division mathematics courses at Universidad Nacional de Colombia, Sede Medell\'in (National University of Colombia at Medell\'in). The Institution is a branch of the National University of Colombia, the best ranked higher education Institution in Colombia; it is divided in five colleges: Architecture, Science, Humanities \& Economical Sciences, Agriculture and Engineering (Facultad de Minas). The colleges are divided in schools and/or departments. The University offers 27 undergraduate programs and 85 graduate programs divided in Specializations, MSc and PhD levels, depending on the school/department. Each semester, the University has an average enrollment of 10000 undergraduates and 2000 graduates with graduation rates of 1240 and 900 respectively. The College of Engineering is the most numerous, consequently, the mathematics lower division courses are highly demanded and have a profound impact in the campus' life; its teaching and evaluation is in charge of the School of Mathematics.

The School of Mathematics is part of the College of Science, it teaches two types of courses: specialization (advanced undergraduate and graduate courses in mathematics) and service (lower division). The latter are: \textit{Basic Mathematics} (BM, college algebra), \textit{Differential Calculus} (DC), \textit{Integral Calculus} (IC), \textit{Vector Calculus} (VC), \textit{Differential Equations} (ODE), \textit{Vector \& Analytic Geometry} (VAG), \textit{Linear Algebra} (LA), \textit{Numerical Methods} (NM), \textit{Discrete Mathematics} and \textit{Applied Mathematics}. The total demand of these courses amounts to an average of 7200 enrollment registrations per semester. The last two courses, \textit{Discrete Mathematics} and \textit{Applied Mathematics} have very low an unstable enrollment, therefore, their data are not suitable for statistical analysis and they will be omitted in the following. On the other hand, the remaining courses are ideally suited for big data analysis, due to its massive nature; see Table \ref{Tb Historical Enrollment Table} below.
\begin{table}[h!]
\def\arraystretch{1.4}
\small{
\begin{center}
\rowcolors{2}{gray!25}{white}
\begin{tabular}{ c c c c c c c c c c c}
    \hline
    \rowcolor{gray!80}
Year 
& Semester
& DC
& IC
& VC
& VAG
& LA
& ODE
& BM
& NM 
& Total\\ 
    \toprule
2010 &	1 &	1631 &	782	  & 381	& 1089	 & 983	& 668 &	848	& 142 & 6524 \\
2013 &	2 &	1446 &	1212 &	549	& 1187	& 1103 & 786 &	846	& 326 & 7455 \\
2016 &	2 &	1569 &	1296 &	594	& 1355	& 1009 & 1019 &	1111 &	284	& 8237 \\
\rowcolor{gray!80}
\hline
Mean & Does not apply & 1445.9	& 1122.9 & 549.8 & 1146.5 & 988.7 & 801.8 & 905.4	& 267.5 & 7228.5 \\
    \hline
\end{tabular}
\end{center}
}
\caption{Historical Enrollment Sample. The table shows the mean enrollment for each course for the period from 2010-1 to 2017-1, together with a sample of three semesters in that time window.}\label{Tb Historical Enrollment Table}
\end{table}
On a typical semester these courses are divided in sections (between 8 to 22, depending on the enrollment) of sizes ranging from 80 to 140 (because of classroom physical capacities). There is no graded homework but students have problem sheets as well as optional recitation classes. As for the \textbf{grading scale} 5.0 is the maximum, 3.0 is the \textbf{minimum pass grade} and grades contain only \textbf{one decimal}. The evaluation consists in three exams which the students take simultaneously; the activity is executed with the aid of the software packages SiDEx-$ \Omega $ and RaPID-$ \Omega$ which manage the logistics of the evaluation and proctoring activities, including the organization of the grading stage. More specifically, for fairness and consistency of the grading process a particular problem is graded by one single grader for all the students, i.e., it is a \textbf{centralized} grading process. As a consequence of the institutional policies described above, the data are statistically comparable. Moreover, the paper-based tests administrator SiDEx-$ \Omega $ introduces high levels of fraud control, because of its \textit{students seating assignment algorithm}; this increases even more the \textbf{reliability} of the data.
%
%
%
%
\subsection{The Databases}
The University allowed limited access to its data bases for the production of this work. The information was delivered in five separate tables which were merged in one single database using Pandas: the file \textit{Assembled\_Data.csv} which contains 108940 rows, each of them with the following fields
\begin{enumerate}[(i)]
\item Student's Personal Information: 
$ \bullet $ Year of Birth
$ \bullet $ e-mail
$ \bullet $ ID Number
$ \bullet $ Last Name and Names
$ \bullet $ Gender

\item Student's General Academic Information:
$ \bullet $ University Entrance Year (AA, Academic Age)
$ \bullet $ Career
$ \bullet $ Academic Average (GPA)

\item Student's Academic Information Relative to the Course:
$ \bullet $ Course
$ \bullet $ Course Code
$ \bullet $ Academic Year 
$ \bullet $ Academic Semester
$ \bullet $ Grade
$ \bullet $ Completed vs Canceled 
$ \bullet $ Number of Attempts
$ \bullet $ Number of Cancellations

\item Student's Administrative Information Relative to the Course:
$ \bullet $ Section Number
$ \bullet $ Schedule
$ \bullet $ Section Capacity
$ \bullet $ Number of Enrolled Students
$ \bullet $ Instructor's ID Number
$ \bullet $ Instructor's Name
$ \bullet $ Tenured vs. Adjoint Instructor

\end{enumerate}
\begin{remark}[Meaning of a row]\label{Rem Row Meaning}
It is understood that one registration corresponds to one row, for instance if a particular student registers for DC and LA in the same term, one row is created for each registration, repeating all the information listed in items (i) and (ii) above. The same holds when an individual needs to repeat a course because of previous failure or cancellation.  
\end{remark}
%
%
%
\section{Quantification: Variables, Segmentation and Profiling}\label{Sec Quantification: Variables, Segmentation and Profiling}
%
%
%
%
In the present work we postulate the Lecturer as one of the most important factors of success, more specifically the aim is to attain the optimal Instructor-Student partnership; in this we differ from \cite{DeGiorgiPellizaariWoolstonGui, DeGiorgiPellizaariRedaelli} where its proposed that the class composition should be driven by the peers interaction. To that end, it becomes necessary to profile the students' population according to its relevant features. 
%
%
%
%
\subsection{Determination of the Segmentation Factor}
Computing the correlation matrix of the quantitative factors considered in the \textit{Assembled\_Data.csv}; in the table \ref{Tb Correlations DC} we display the correlation matrix for the Differential Calculus course. From the \textit{Grade} row it is clear that the most significant factor in the \textit{Grade} variable is the \textit{Academic Average} (GPA) followed by the \textit{Academic Age} (AA) and the \textit{Age}. However, the impact of the GPA is about four times the impact of AA and the same holds for the \textit{Age} factor. Moreover, from the GPA row, it is clear that the most significant factor after the \textit{Grade} are precisely the \textit{Academic Age} and the \textit{Age} (the younger the student, the higher the GPA). In addition, for the remaining courses, similar correlation matrices are observed. Hence, we keep the GPA as the only significant quantitative factor in the \textit{Grade} variable. 
\begin{remark}\label{Rem sections size}
It is important to mention that the impact of the class size in the students' performance has been subject of extensive discussion without consensus. While 
\cite{AngrisLavy, Kruegger} report a significant advantage in reducing class sizes, \cite{Hoxby} finds no effect. In our particular case \textsc{Table} \ref{Tb Correlations DC} shows that the \textit{Section Capacity} is uncorrelated, not only to the \textit{Grade} variable, but also to the \textit{GPA} variable. Moreover the \textit{Section Capacity} is uncorrelated with the \textit{Cancellations} (drop out) variable.   
\end{remark}
\begin{table}[h!]
\def\arraystretch{1.4}
\scriptsize{
\begin{center}
\rowcolors{2}{gray!25}{white}
\begin{tabular}{ c | c c c c c c c c }
    \hline
    \rowcolor{gray!80}
FACTOR
& Section 
& Age
& AA
& \# Enrolled 
& Grade
& \# Cancellations
& \# Attempts
& GPA
\\
\rowcolor{gray!80}
& Capacity
&
& 
& Students
& 
& 
& 
& 
\\
\hline 
Section Capacity
& 1.0000 	 & -0.0108	& 0.0570	& 0.8334    & 0.0074	& 0.0003  &	0.0393 & 0.0180 \\
Age
& -0.0108	& 1.0000	& 0.3069	& -0.0168	& -0.2031  & 0.0775	 & 0.1384 &	-0.2082 \\
AA
& 0.0570	 & 0.3069	 & 1.0000	 & 0.0416	 & -0.2164	& 0.1668  &	0.4294 &	-0.1667 \\
\# Enrolled Students
& 0.8334	& -0.0168   & 0.0416	& 1.0000	& 0.0252    & 0.0131  &	-0.0041	& 0.0325 \\
\ding{217} Grade
& 0.0074	& -0.2031	& -0.2164  & 0.0252	   & 1.0000	   & -0.1247 &	-0.0241	& $ \star $ 0.8207 \\
\# Cancellations
& 0.0003	& 0.0775	& 0.1668	& 0.0131   	& -0.1247	& 1.0000   &	0.3101 & -0.0686 \\
\# Attempts
& 0.0393	& 0.1384	& 0.4294	& -0.0041  & -0.0241   &	0.3101 &	1.0000 &	-0.0401 \\
\ding{217} GPA 
& 0.0180	& -0.2082	& -0.1667 &	0.0325	  & $ \star $ 0.8207	   & -0.0686  &	-0.0401	& 1.0000 \\
\hline
\end{tabular}
\end{center}
}
\caption{Quantitative Factors Correlations Table, Course: \textbf{Differential Calculus}. The table displays the correlation matrix for the variables of interest (Section Capacity, Age, AA, \# Enrolled Students, Grade, \# Cancellations, \# Attemps, GPA).
}\label{Tb Correlations DC}
\end{table}

Two binary variables remain to be analyzed namely  \textit{Pass/Fail} (PF) and \textit{Gender}. If we generically denote by $ X $  the binary variables and by $ Y $ a variable of interest, the point-biserial correlation coefficient is given by
\begin{equation}\label{Eq point-biserial correlation}
r_{pb} \defining \frac{M_{1} - M_{0} }{\sigma} \sqrt{\frac{N_{1}\, N_{0}}{N^{2}}}.
\end{equation}
Here, the indexes $ 0,1 $ are the values of the binary variable $ X $. For $ i = 0, 1 $, $ M_{i} $ is the mean value of the variable $ Y $ for the data points in the group/event $ \{ X = i \} $, $ N_{i} $ denotes the population of each group $ \{ X = i \} $, $ N = N_{1} + N_{0} $ stands for the total population and $ \sigma $ indicates the standard deviation of the variable $ Y $. 
%
%

The correlation analysis between the binary \textit{Pass/Fail} (PF) variable vs the quantitative factors is displayed in Table \ref{Tb PF Biserial Correllations}. As in the \textit{Grade} variable analysis, the most significant factor in the \textit{Pass/Fail} variable, is the \textit{Academic Average} (GPA) followed by the \textit{Academic Age} (AA) and the \textit{Age}. In this case however, the impact of the GPA is only three times the impact of AA as well as the \textit{Age} factor. Again, we keep the GPA as the only significant quantitative factor in the \textit{Pass/Fail} variable. 
\begin{table}[h!]
\def\arraystretch{1.4}
\small{
\begin{center}
\rowcolors{2}{gray!25}{white}
\begin{tabular}{ c | c c c c c c c c }
    \hline
    \rowcolor{gray!80}
\diagbox{FACTOR}{COURSE}
& DC
& IC
& VC
& VAG
& LA
& ODE
& BM
& NM
 \\
\hline
Section Capacity 
& 0.0058 &	0.0139 & -0.0296 & -0.0986 & -0.0147 &	0.0448 & -0.0451 & 0.0840 \\
Age
& -0.1757 &	-0.2554 & -0.3048 &	-0.1879	& -0.2444 &	-0.2555 & -0.0892 & -0.3333 \\
AA
& -0.1783 &	-0.2855 & -0.3012 &	-0.1550 & -0.2230 &	-0.3841	& -0.0209 &	-0.2731 \\
\# Enrolled Students
& 0.0070 &	0.0276 & 0.0172	& -0.0410 & 0.0602 & 0.1258	& -0.0374 &	0.1916 \\
Grade
& 0.8072 & 0.7987 &	0.7864 & 0.8145	& 0.7959 & 0.7999 &	0.7988 & 0.7922 \\
\# Cancellations
& -0.1129 &	-0.1397 & -0.1420 &	-0.1299	& -0.1178 &	-0.1647	& -0.0096 &	-0.1489 \\
\# Attempts
& -0.0988 &	-0.1635	& -0.1683 &	-0.1374	& -0.1399 &	-0.1577 & -0.0054 &	-0.2065 \\
\ding{217} GPA
& 0.6062 & 0.5892 &	0.5884 & 0.6445 & 0.6063 & 0.5341 &	0.6125 & 0.5828
\\
\hline
\end{tabular}
\end{center}
}
\caption{Pass/Fail vs Quantitative Factors, Biserial Correlations Table, 
	Course: \textbf{All}. We display the biserial correlations for each of the variables (Section Capacity, Age, AA, \# Enrolled Students, Grade, \# Cancellations, \# Attemps, GPA) compared to the Pass/Fail variable, per course.}\label{Tb PF Biserial Correllations}
\end{table}

Next, the correlation analysis \textit{Gender} variable vs the Academic Performance Variables, i.e., \textit{Grade}, GPA and \textit{Pass/Fail} (PF) is summarized in the table \ref{Tb Gender Correllations} below. Clearly, the \textit{Gen}der variable has negligible incidence in the Academic Performance Variables, with the exception of the GPA for the Basic Mathematics (BM) course case, where females do slightly better. Since this unique correlation phenomenon is not present in the remaining courses, the \textit{Gender} variable will be neglected from now on.  Finally, it is important to stress that, given the binary nature of the \textit{Gender} and the \textit{Pass/Fail} (PF) variables, all the correlation coefficients agree i.e., point-biserial, Pearson and Spearman and Kendall. 
\begin{table}[h!]
\def\arraystretch{1.4}
\small{
\begin{center}
\rowcolors{2}{gray!25}{white}
\begin{tabular}{ c | c c c c c c c c }
    \hline
    \rowcolor{gray!80}
\diagbox{FACTOR}{COURSE}
& DC
& IC
& VC
& VAG
& LA
& ODE
& BM
& NM
 \\
\hline
Grade
& 0.0545 &	0.0584 & 0.0240	& -0.0485 &	0.0154 & 0.0164	& -0.0740 &	-0.0405 \\
GPA
& -0.0425 &	-0.0441	& -0.0419 &	-0.0766	& -0.0343 & -0.0824	& -0.1159 & -0.0663 \\
Pass/Fail (PF)
& 0.0509 &	0.0382	& 0.0166 &	-0.0373	& 0.0103 &	0.0085 & -0.0535 & -0.0226 \\
\hline
\end{tabular}
\end{center}
}
\caption{Gender vs Academic Performance Variables, 
	Course: \textbf{All}. The biserial correlations of each of the variables in (Grade, GPA, Pass/Fail) compared to the Gender variable, is displaye for each course.}\label{Tb Gender Correllations}
\end{table}
%
%

From the previous discussion, it is clear that out of the analyzed variables, the GPA is the only one with significant incidence on the academic performance variables \textit{Grade} and \textit{Pass/Fail}. Consequently, this will be used as the unique criterion for the segmentation of students' population. 
%
%
From now on, our analysis will be focused on the \textit{Grade Average} and the \textit{Pass/Fail} variables as measures of success, while the \textit{GPA} will be used for segmentation purposes discussed in Section \ref{Sec Segmentation Process}. In \textsc{Table} \ref{Tb Academic Performance Variables Average}, the global averages (from 2010-1 to 2017-1) of these variables for all the service courses are displayed. 
\begin{table}[h!]
\def\arraystretch{1.4}
\small{
\rowcolors{2}{gray!25}{white}
\begin{center}
\begin{tabular}{ c | c c c c c c c c }
    \hline
    \rowcolor{gray!80}
\diagbox{VARIABLE}{COURSE}
& DC
& IC
& VC
& VAG
& LA
& ODE
& BM
& NM
 \\
\hline
Grade &
2.6849	& 2.7829 &	3.2198	& 2.8616 &	3.0233 & 3.1170 & 2.8308 &	3.1893\\
GPA &
3.2213 & 3.3527	& 3.4969 & 3.2386 & 3.3201	& 3.4548 &	3.2696 & 3.5330 \\
Pass/Fail &
0.5010 & 0.5339 & 0.7151 & 0.5901 &	0.6398	& 0.6846 &	0.5441	& 0.6924 \\
Number of Tries &
1.7382 & 1.9140	& 1.4996 & 1.5205 &	1.5495	& 1.7710 &	1.0549	& 1.4115 \\
\hline
\end{tabular}
\end{center}
}
\caption{Academic Performance Variables Average Values, 
	Course: \textbf{All}. The average performance (as captured by the variables Grade, GPA, and Pass/Fail) of students from 2010-1 to 2017-1, is displayed for each course.}\label{Tb Academic Performance Variables Average}
\end{table}
%
%
%
%
\subsection{Segmentation Process}\label{Sec Segmentation Process}
The profiling of students' population is to be made course-wise. For each group taking a course, the algorithm computes a partition of the interval $ [0,5] $ of ten numerical GPA intervals $\big(I_{\ell}: \ell \in [10] \big)$, so that approximately ten percent of the population is contained in $ I_\ell $ for all $ \ell \in [10] $. Equivalently, if a histogram of relative frequencies is drawn, as in \textsc{Figure} \ref{Fig DC GPA Histograms}, the area between the curve and any interval should be around 0.1. Hence, if $ f_{\gpa}  $ is the relative frequency of the $ \gpa $ variable then $ \int_{I_{\ell} } f_{\gpa} \, dx \sim 0.1 $ for all $ \ell \in [10] $. The process described above is summarized in the pseudocode \ref{Alg Segmentation of Students} below.

\begin{figure}[h!]
        \centering
        \begin{subfigure}[Example DC. GPA Histogram, Semesters from 2010-1 to 2017-1. ]
{\includegraphics[scale = 0.38]{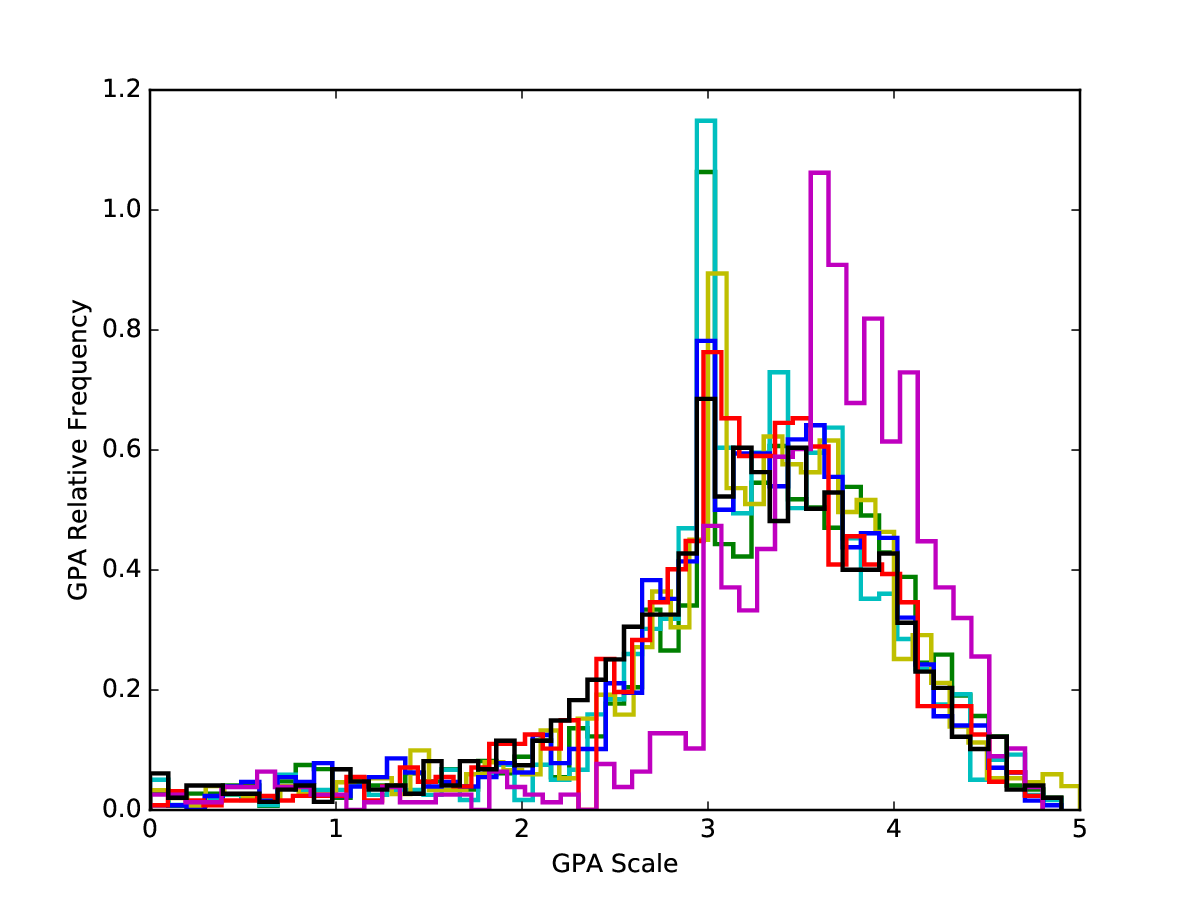} } 
        \end{subfigure}
        ~ 
          \begin{subfigure}[Example DC. GPA Histogram, Semesters 2010-1 and 2015-1.]
{\includegraphics[scale = 0.38]{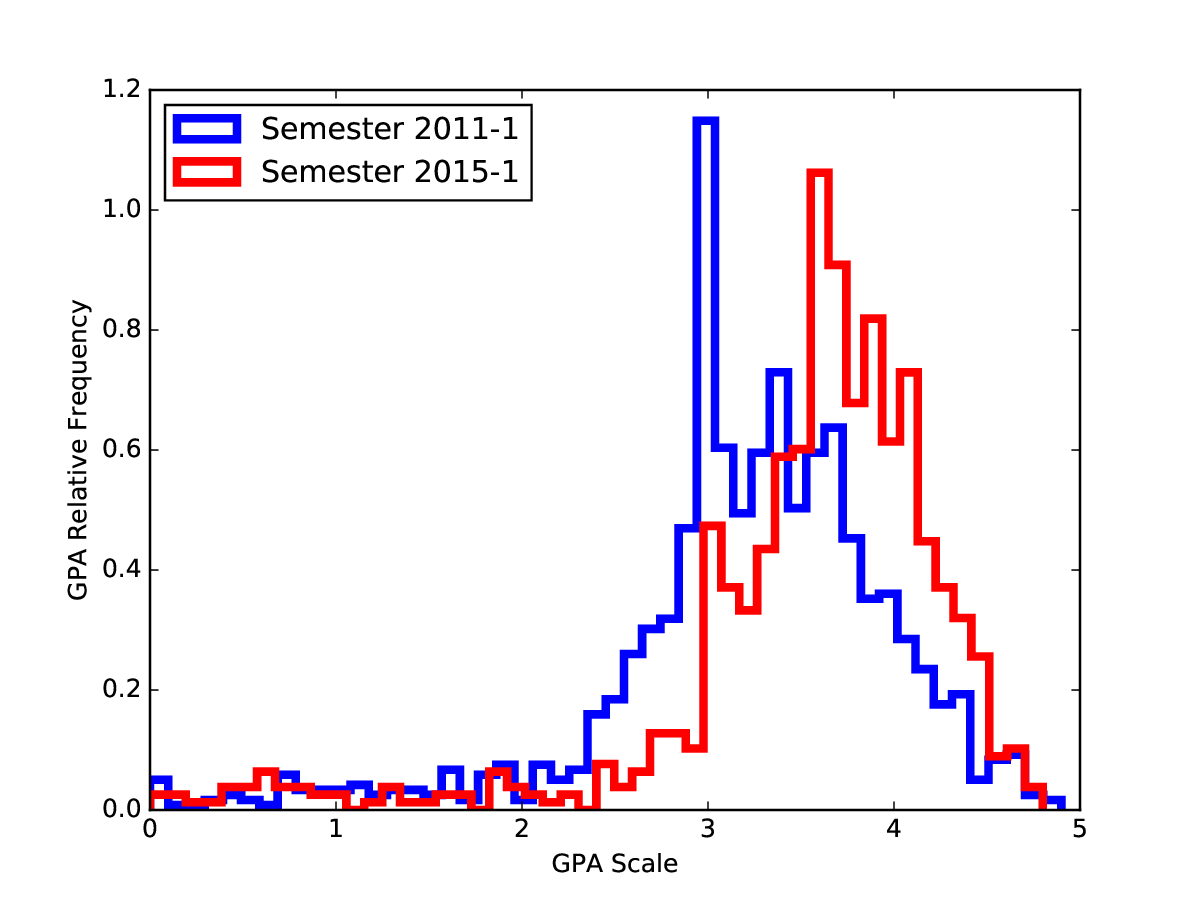} } 
        \end{subfigure}
\caption{Differential Calculus GPA Normalized Histogram. Figure (a) Shows the normalized frequencies histograms for all the semesters available in the database \textit{Assembled\_Data.csv}. Figure (b) displays the normalized histogram of only two semesters for optical purposes. Observe that in both cases the area beneath the normalized histogram is exactly equal to one. 
\label{Fig DC GPA Histograms} }
\end{figure}
\begin{algorithm}[H]
\KwData{
Database: Assembled\_Data.csv. \\
Year ($ y $) and Semester ($ s $). \\
\textit{Analyzed Course}: DC, IC, ..., NM. 
}
\KwResult{Extremes of the GPA Segmentation Intervals $ \big(I_{\ell}: \ell \in [L]\big) $; extremes = [0, n\_1, n\_2,\ldots, 5 ]}
\textbf{Initialization}\;
List\_GPA $ \leftarrow $ \textbf{sorted} (\textbf{hash} $ \gpa $ variable from \\
Assembled\_Data.csv[(Course = \textit{Analyzed Course}) \& (Year = $ y $) \& (Semester = $ s $ )] ) \;
extremes\_0 $ \leftarrow $ 0\; 
\For{ $ i \in [10] $ }{
extremes\_i $ \leftarrow \big\lfloor i\times  \frac{\text{length of List\_GPA} }{10} \big\rfloor$\;
}
\If{list {extremes} contains repetitions}{
extremes $ \leftarrow $ remove repetitions from extremes}
\caption{Segmentation of Students}\label{Alg Segmentation of Students}
\end{algorithm}
\begin{remark}\label{Rem Segmentation Algorithm}
Observe that Algorithm \ref{Alg Segmentation of Students} is aimed to produce ten segmentation intervals, however the last instruction considers removing some points out of the eleven extremes previously defined, in case of repetition. Such situation could arise when a particular GPA value is too frequent as it can be seen in \textsc{Figure} \ref{Fig DC GPA Histograms} (b), for the case of Semester 2011-1, which has a peak at GPA = 3. Similar peaks can be observed in other semesters as \textsc{Figure} \ref{Fig DC GPA Histograms} (a) shows.
\end{remark}
%
%
%
%
\subsection{Computation of the Lecturer's Performance}\label{Sec Computation Lecturer Performance}
%
%
%
%
The treatment of the lecturer as a success factor is completely tailored to the case of study and it can not be considered as a general method, the expected (average) performance will be computed for the \textit{Grade} and the \textit{Pass/Fail} variables. For the computation of instructors' performance, first a segmentation process $ \big(I_{\ell}: \ell \in [L]\big) $ (as described in Subsection \ref{Sec Segmentation Process}) has to be done. Next, the computation is subject to the following two principles 
\begin{enumerate}[(i)]
\item Adjunct and Tenured (Track or not) lecturers are separated in different groups. 

\item If the experience of a particular instructor (the full personal teaching log inside the database \textit{Assembled\_Data.csv}) within a segment $ I_{\ell} $ of analysis, accumulates less than 30 individuals, his/her performance within such $ I_{\ell} $ is replaced by the average performance of the group he/she belongs to (Adjunct or Tenured) within such group, i.e., the conditional expectation of the Academic Performance Variable ($ \apv $) subject to the segment of analysis: $ \Exp\big( \apv \big\vert \text{Instructor} = x\big) $, see \cite{MendenhallBeaver, BillingsleyProb}.
\end{enumerate}
\begin{remark}\label{Rem Adjunct lecturers treatment}
The separation of Adjunct and Tenured instructors is done because the working conditions, expectations, as well as the results, are significantly different from one group to the other inside the Institution of analysis. In particular, the adjunct instructors are not stable nor full-time personnel. Consequently, these two groups are hardly comparable. On the other hand, there is an internal policy of rotating teaching faculty through the lower division courses, according to the needs of the School of Mathematics. Hence, due to the hiring and teaching-rotation policies, an Adjunct instructor rarely accumulates 30 or more students of experience within a profile segment $ I_{\ell} $.  
\end{remark}
%
%
%
%
\begin{remark}[Economy Perspective]\label{measurament_of_instructor}
Measuring instructor performance through the students \textit{Grade} and \textit{Pass/Fail} variables, treating instructors as a transformation function in which output (student results) is measured with respect to the input (students background), was the norm in the past (see, e.g., \cite{higgins1989performance}). Nonetheless this approach has several problems, as pointed out in \cite{higgins1989performance} and \cite{rockoff2010subjective}. Some of these problems are: the difficulty in accurately measuring students' background, the existing bias charged on instructors tasked with students who are more difficult to teach, and the non-comparability of students' grades across different instructors. Nonetheless, time and again, new developments on how to measure instructors' performance appear. In \cite{berk2005survey}, R. A. Berk presents 12 strategies to measure teaching effectiveness, some of these measures are: peer ratings, self-evaluation, alumni ratings, teaching awards and others. Finally, given the algorithmic nature of our work we only need one variable of instructor performance in order to present the optimization method, but the algorithm itself applies to any quantitative measure, as the ones just mentioned above. 
\end{remark}
%
%
The performance computation is described in the following pseucode

\begin{algorithm}[H]
\KwData{
Database: Assembled\_Data.csv. \\
\textit{Analyzed Course}: DC, IC, ..., NM.\\
Academic Performance Variable ($ \apv $): Grade, Pass/Fail \\
Group Segmentation $ \big( I_{\ell} : \ell \in [L]\big) $. }
\KwResult{Hash tables of performance for the analyzed course, conditioned to each segmentation\\ interval $ \big( I_{\ell}: \ell \in [L]\big) $ and for the chosen Academic Performance Variable ($ \apv $): 
APV\_Performance\_Tenured$ [\ell] $,
APV\_Performance\_Adjoint$ [\ell] $,
APV\_Performance\_Instructor$ [\ell] ,\, \ell \in [L]$}
\textbf{Initialization}\;
Instructors\_List $ \leftarrow $ \textbf{hash} lecturer list from Assembled\_Data.csv[Course = \textit{Analyzed Course}] \; 
Tenured\_List $ \leftarrow $ \textbf{choose} from Instructors\_List the Tenured lecturers\; 
Adjoint\_List $ \leftarrow $ Instructors\_List $ - $ Tenured\_List\; 

\For{ $ \ell \in [L] $ }{
X = \textbf{hash} $ \apv $ field from table:\newline
Assembled\_Data.csv[(Course = \textit{Analyzed Course}) \& (GPA $ \in  I_{\ell} $) \& (\textit{Instructor} $ \in $ Tenured\_List) ]\;
APV\_Performance\_Tenured$ [\ell]\leftarrow \Exp(\X) $ \;
X = \textbf{hash} $ \apv $ field from table: \newline
Assembled\_Data.csv[(Course = \textit{Analyzed Course}) \& (GPA $ \in  I_{\ell} $) \& (\textit{Instructor} $ \in $ Adjoint\_List) ]\;
APV\_Performance\_Adjoint$ [\ell] \leftarrow \Exp(\X) $.
}
\For{ $ \ell \in [L] $ }{
	\For{instructor $ \in $ Instructors\_List}{
	X = \textbf{hash} $ \apv $ field from table:\newline
	Assembled\_Data.csv[(Course = \textit{Analyzed Course}) \& (GPA $ \in  I_{\ell} $) \& (\textit{Instructor} = instructor) ]\;
		\eIf{length of $ \X  >=  30 $ }{
		APV\_Performance\_Instructor$ [\ell]\leftarrow \Exp(\X) $ .}
		{
			\eIf{instructor $ \in $ Tenured\_List}{
			APV\_Performance\_Instructor$ [\ell]\leftarrow $APV\_Performance\_Tenured$ [\ell] $ }
			{APV\_Performance\_Instructor$ [\ell]\leftarrow $APV\_Performance\_Adjoint$ [\ell] $ }
		}
}
}
\caption{Computation of Instructors' Performance}\label{Alg Instructors' performance}
\end{algorithm}
%
%
%
%
%
%
%
%
\section{Core Optimization Algorithm and Historical Assessment}\label{Sec Optimization and Historical Assessment}
%
%
%
%
%
%
In this section we describe the core optimization algorithm. Essentially, it is the integration of the previous algorithms with an integer programming module whose objective function is to maximize the \textit{Expectation} of the academic performance variables (\textit{Grade} and \textit{Pass/Fail}), according to the big data analysis described in \textsc{Section} \ref{Sec Computation Lecturer Performance}. Two methods are implemented for each course and semester recorded in the database.
\begin{enumerate}[I.]
\item \textbf{Instructors Assignment (IA).} Assuming that the groups of students are already decided, assign the instructors pursuing the optimal expected performance partnership: Instructor-Conformed Group. This is known in integer programming as the Job Assignment Problem.

\item \textbf{Students Assignment (SA).} Assuming that the sections (with a given capacity) and their corresponding lecturers are fixed, assign the students to the available sections in order to optimize the expected performance of the Student-Instructor partnership. This is the integer programming version of the Production Problem in linear optimization. 
\end{enumerate}   
In order to properly model the integer programs we first introduce some notation
\begin{definition}\label{Def Characteristic Numbers}
Let $ N, L, J \in \N $ be respectively the total number of students, the total number of segmentation profiles and the total number of sections. Let $ \p = \big(p_{\ell}: \ell \in [L] \big) \in \N^{L} $, $ \g = \big(g_{j}: j \in [J] \big) \in \N^{J} $ be respectively the population of students in each profiling segment and the capacities of each section, in particular the following sum condition holds.
\begin{equation}\label{Eq Sum Condition}
\sum_{\ell \, = \, 1}^{L} p_{\ell} =   \sum_{j \, = \, 1}^{J} g_{j} = N .
\end{equation}
\end{definition}
\begin{remark}\label{Rem Salck Variables Absence}
Observe that the condition $ \sum_{j \, = \, 1}^{J} g_{j} = N  $ implies there are no slack variables for the capacity of the sections. This is due to the study case, in contrast with other Universities where substantial slack capacities can be afforded.
\end{remark}
\begin{definition}\label{Def Performance and Group Assignment Matrices}
Let $ N, L, J \in \N $ , $ \p \in \N^{L}, \g \in \N^{J} $ be as in Definition \ref{Def Characteristic Numbers}. 
\begin{enumerate}[(i)]
\item 
We say a matrix $ G \in \R^{L\times J} $ is a \textbf{group assignment matrix} if all its entries are non-negative integers and 
\begin{align}\label{Eq Row Sum Column Sum Condition}
& \sum\limits_{j \,\in\,[J] } G(\ell,j) = p_{\ell}, \;  \forall \, \ell \in [L] , &
& \sum\limits_{\ell \,\in\,[L] } G(\ell,j) = g_{j}, \; \forall \,  j \in [J] .
\end{align}
Furthermore, define the \textbf{group assignment space} by $ \itG \defining \big\{ G: G \text{ is a group assignment matrix} \big\} $.  

\item Let $ \big(t_{j}: j\in [J]\big) $ be the intructors assigned to the course. For a fixed $ \apv \in \big\{\text{Grade, Pass/Fail} \big\} $ define \textbf{expected performance matrix} $ T_{\apv} \in \R^{J \times L} $ as the matrix whose entries $ T_{\apv}(j, \ell) $ are the $ \apv $ variable performance, corresponding to the instructor $ t_{j} $ within the segmentation interval $ I_{\ell} $.

\item Given a group assignment matrix $ G $ and a faculty team $ \big(t_{j}: j\in {J}\big) $, define the \textbf{choice performance matrix} $ C_{\apv} $ by
\begin{equation}\label{Eq Costs Table APV}
C_{\apv} \defining T_{\apv} G .
\end{equation}

\end{enumerate}
\end{definition}
\begin{remark}\label{Rem Group Assignment Matrix and Performance Matrix}
	Observe the following
	\begin{enumerate}[(i)]
		\item $ C_{\apv}(j,i) $ measures the average performance of the instructor $ t_{ j } $  over the partition $ \{G(k,i)\}_{k=1}^{L} $ of the section $ g_{ i } $. 
		
		\item Recall from combinatorics that a weak composition of $ n $ in $ m $ pars is a sequence of inon-negative ntegers  $ (a_{ 1 }, \ldots, a_{ m } ) $  satisfying $ \sum_{ i \, = \, 1 }^{ m } a_{ i } = n $ (see \cite{BonaWalk}). Notice that $ \big(G(\ell,j): j\in [J]\big) $ is a weak composition of $ p_{\ell} $ for every $ \ell\in [L] $ and that $ \big(G(\ell,j): \ell\in [L]\big) $ is a weak composition of $ g_{j} $ for every $ j\in [J] $.

\item Recall that the expected performance matrix $ T_{\apv} $ for $ \apv \in \big\{\text{Grade, Pass/Fail} \big\} $, is recovered from the file APV\_Performance\_Instructor hash table constructed in Algorithm \ref{Alg Instructors' performance}, Section \ref{Sec Computation Lecturer Performance}. 
	\end{enumerate}
\end{remark}
Next we introduce the integer problems
\begin{problem}[Instructors Assignment Method]\label{Pblm Instructors Assignment}
Let $ N, L, J \in \N $ be as in \textsc{Definition} \ref{Def Characteristic Numbers}, let $ \xi = \big(\xi(i, j): i \in[I], j\in [J] \big) \in \big\{0,1\big\}^{L\times J} $ and let $ C_{\apv} $ be as in \textsc{Definition} \ref{Def Performance and Group Assignment Matrices} for a fixed group assignment matrix $ G $ and faculty team $ \big(t_{j}: j\in {J}\big) $. Then, the instructors assignment problem is given by
\begin{subequations}\label{Eq Instructors Assignment}
\begin{equation}\label{Eq Instructors Assignment Objective Function}
v_{\ia} \defining \max\limits_{ \xi \, \in  \, \{0,1\}^{L\times J} } 
\sum_{i\, = \, 1}^{J} \sum_{j\, = \, 1}^{J} C_{\apv}(i, j) \, \xi(i, j) 
,
\end{equation}
subject to:
\begin{align}\label{Eq Instructors Assignment Constraints}
& \sum_{i\, = \, 1}^{J} \xi(i, j)  = 1, \; \forall\, j \in [J] , & 
& \sum_{j\, = \, 1}^{J}  \xi(i, j) = 1 ,  \; \forall\, i \in [J] .
\end{align}
\end{subequations}
\end{problem}
\begin{problem}[Students Assignment Problem]\label{Pblm Students Assignment Problem}
With the notation introduced in \textsc{Definition} \ref{Def Characteristic Numbers} and a chosen faculty team $ \big(t_{j}: j\in {J}\big) $, let $ \pi $ be a permutation in $ \itS_{J} $ such that $ t_{\pi(j)} $ is the instructor of section $ j $ for all $ j \in [J] $ i.e., a chosen assignment of lecturers to the sections. Then, the students assignment problem is given by
\begin{equation}\label{Eq Students Assignment Problem}
v_{\sa} = \max\limits_{G\, \in\, \itG } 
\sum_{j\, = \, 1}^{J} \big(T_{\apv} G \big)\big(j, \pi(j) \big) 
=  \max\limits_{G\, \in\, \itG } 
\sum_{j\, = \, 1}^{J} C_{\apv}\big(j, \pi(j) \big) 
.
\end{equation}
\end{problem}
\begin{remark}\label{Rem Integer Programming Problems}
\begin{enumerate}[(i)]
\item Observe that the constraints of the problem \ref{Pblm Students Assignment Problem} are only those of \textsc{Equation} \eqref{Eq Row Sum Column Sum Condition}; these are fully contained in the condition $ G \in \itG $.

\item Notice that although the search space of \textsc{Problem} \ref{Pblm Students Assignment Problem} is significantly bigger than the search space of \textsc{Problem} \ref{Pblm Instructors Assignment}, the optimum of the former need not be bigger or equal than the optimum of the latter. However, in practice, the numerical results below show that this is the case, not because of search spaces inclusion, but due to the overwhelming difference of sesarch space sizes.
\end{enumerate}
\end{remark}
In order to asses the enhancement introduced by the method, it is necessary to compute rates of optimal performance over the historical one i.e., if $ G_{h} \in \itG $, $ \pi_{h} \in \itS_{J} $ are respectively the historical group composition and instructors assignation for a given semester $ h $ then, the relative enhancement $ \rho_{\mt} $, due to a method $ \mt $ is given by 
\begin{align}\label{Eq Performance Rates}
& \rho_{\mt} \defining 
100\frac{v_{\mt} - \sum\limits_{j\, = \, 1}^{J} \big(T_{\apv} G_{h}\big) \big(j, \pi_{h}(j) \big) }
{\sum\limits_{j\, = \, 1}^{J} \big(T_{\apv} G_{h}\big) \big(j, \pi_{h}(j) \big) }
\, , &
& \mt \in \{ \ia, \sa\} .
\end{align} 
Finally, we describe in \textsc{Algorithm} \ref{Alg Optimization Algoritm} below the optimization algorithm
\begin{remark}[Economy Perspective: $\ia $ and $ \sa $ solutions as Pareto equilibria]\label{An Pareto Balance}
\begin{enumerate}[(i)]
\item The $ \ia $ and $ \sa $ problems are two  scheduling formulations driven by social welfare. The University as a central regulator agent aims to solve such scheduloing problems in order to improve the social welfare of its community (i.e, students and professors). Given a group assignment matrix $G$, the $ \ia $ method seeks to find the matching pairs $(\mathit{instructor}, \mathit{section})$ in order to maximize the total average performance of the instructors, subject to the constraint that each instructor must teach only one section. On the other hand, the $ \sa $ method seeks to find a group assignment matrix $G$, given a complete matching pairs $(\mathit{instructor},\mathit{section})$.  More specifically, a distribution of the students population that maximizes the total average performance.\footnote{ Notice that the $ \ia $ method is easire to implement than the $ \sa $ method, the former only requires to allocate the instructors, while the later requires a redistribution of the whole students population. } 

\item 
The solutions $ v_{\ia} $ and $ v_{\sa} $ are configurations corresponding to Pareto equilibria, i.e. situations where no individual can improve his/her welfare (success chances in this particular case) without decreasing the well-being of another individual of the system. In this same spirit, the parameters $ \rho_{\mt} $ are measures of deviation from the Pareto equilibrium. 
\end{enumerate}
\end{remark}
\begin{algorithm}[t]
\KwData{
Database: Assembled\_Data.csv. \\
Year ($ y $) and Semester ($ s $). \\
Analyzed Course: DC, IC, ..., NM.\\
Academic Performance Variable ($ \apv $): Grade, Pass/Fail. \\
Group Segmentation $ \big( I_{\ell} : \ell \in [L]\big) $.\\
APV\_Performance\_Instructor$ [\ell] ,\, \ell \in [L]$. \\
Optimization Method: $ \mt \in \{\ia, \sa\} $. }
\KwResult{Relative enhancement value $ \rho_{\mt} $ for method $ \mt $, for chosen course, year and semester. 
}
\textbf{Initialization}\;
Course\_Table $ \leftarrow $ \textbf{hash} \\
Assembled\_Data.csv[(Course = \textit{Analyzed Course}) \& (Year = $ y $) \& (Semester = $ s $) ]\;
Instructors\_List $ \leftarrow $ \textbf{hash} lecturer list from Course\_Table \;
Instructors\_Performance $ \leftarrow $ \textbf{hash} APV\_Performance\_Instructor[Instructor $ \in \text{Instructors\_List}$ ]\;
Section\_List $ \leftarrow $ \textbf{hash} section list from Course\_Table 

\For{ $ \ell \in [L] $ }{
	\For{ $ j \in J $ }{
	$ T_{\apv} (\ell, j) \leftarrow $ \textbf{hash} \\
	APV\_Instructors\_Performance[(Instructor = Instructors\_List$ (j) $ )
	\& (Segmentation = $ I_{\ell} $)]\;
	$ G_{h}(\ell, j) \leftarrow $ lenght(\textbf{hash} Course\_Table[(Section = $ j $) \& (Segmentation = $ I_{\ell} $)])
	}
}
\eIf{$ \mt = \ia $}{
	$ C_{\apv} \leftarrow T_{\apv} G_{h} $,
	$v_{\ia}  \leftarrow $ \textbf{solve} Problem \ref{Pblm Instructors Assignment}, \textbf{input}: $ C_{\apv} $.
	}
	{
	$ \p \leftarrow \Big(\sum_{j\, = \, 1}^{J}G(\ell, j) : \ell \in [L]\Big) $,
	$ \g \leftarrow \Big(\sum_{\ell\, = \, 1}^{L}G(\ell, i) : i \in [J]\Big) $,
	$ \pi \leftarrow $ Section\_List\;
	$v_{\sa}  \leftarrow $ \textbf{solve} Problem \ref{Pblm Students Assignment Problem}, \textbf{input}: $ ( T_{\apv}, \p, \g, \pi) $.
} 
$ \pi_{h} \leftarrow $ Section\_List\;
$ \rho_{\mt} \leftarrow $ \textbf{compute} Equation \eqref{Eq Performance Rates}, \textbf{input}: 
$ \big( T_{\apv} , G_{h}, \pi_{h}, v_{\mt}\big) $.
\caption{Optimization Algorithm}\label{Alg Optimization Algoritm}
\end{algorithm}
%
%
%
\subsection{Historical Assessment}\label{Sec Historical Assessment}
In the current section, we are to assess the enhancement of the proposed method with respect to the average of the historical results. To that end, we merely integrate \textsc{Algorithms} \ref{Alg Segmentation of Students}, \ref{Alg Instructors' performance} and \ref{Alg Optimization Algoritm} in a master algorithm going through a time loop to evaluate the performance of each semester and then store the results in a table, this is done in \textsc{Algorithm} \ref{Alg Analytica Omega}. It is important to observe that excepting for the database, all the remaining input data must be defined by the user.

The numerical results for the Differential Calculus course are summarized in \textsc{Table} \ref{Tb Historical Enhnacements} and illustrated in \textsc{Figure} \ref{Fig Historical Enhancement}. The results clearly show that the Students Assignment method (SA) yields better results than the Instructor Assignment method (IA), which holds for both Academic Performance Variables: \textit{Pass/Fail} and \textit{Average}. Such difference happens not only for the mean value, but on every observed instance (semester), this is due to the difference of size between search spaces for the problems \ref{Pblm Instructors Assignment} and \ref{Pblm Students Assignment Problem} as discussed in Remark \ref{Rem Integer Programming Problems}. On the other hand, it can be observed that the \textit{Pass/Fail} variable is considerably more sensitive to the optimization process than the \textit{Average} variable. Again, the phenomenon takes place not only for the enhancement's mean value, the former around three times the latter, but the domination occurs for every semester analyzed by the algorithm. The latter holds because, for an improvement on the \textit{Average} variable to occur, a general improvement in the students' grades should take place, while the improvement of the pass rate is not as demanding.

The results of the optimization methods yield similar behavior for all the remaining lower division courses. Consequently, in the following we will only be concerned with the analysis of the \textit{Pass/Fail} variable, which gives the title to the present paper. The two optimization methods will be kept for further analysis, not because of efficiencies (clearly SA yields better results), but because of the administrative limitations a Higher Education Institution could face when implementing the solution. Clearly, from the administrative point of view, it is way easier for an Institution implementing IA instead of SA, 
%
%
%
%

It is also important to mention, that although enhancements of 1.4 or 7 percent may not appear significant at first sight, the benefit is substantial considering the typical enrollments displayed in \textsc{Table} \ref{Tb Historical Enrollment Table}, as well as the average \textit{Number of Tries} a student needs to pass de course displayed in \textsc{Table} \ref{Tb Academic Performance Variables Average}. In addition, the fact that Latin American public universities heavily subside its students despite having serious budgeting limitations (as in our study case), gives more relevance to the method's results.   

\begin{algorithm}[H]
\KwData{
Database: Assembled\_Data.csv. \\
Analyzed Course: DC, IC, ..., NM.\\
Academic Performance Variable ($ \apv$ ): Grade, Pass/Fail. \\
Optimization Method: $ \mt \in \{\ia, \sa\} $. }
\KwResult{Table of Relative Enhancement Values $ \rho_{\mt} $ for chosen method, course and academic performance variable.
}
\textbf{Initialization}\;
\For{ Year $ \in [2010, 2017]$ }{
	\For{ Semester $ \in [2] $ }{
	\textbf{call} Algorithm \ref{Alg Segmentation of Students}, \textbf{input}: (\textit{Assembled\_Data.csv}, \textit{Year}, \textit{Semester}, \textit{Analyzed Course})\;
	\textbf{call} Algorithm \ref{Alg Instructors' performance}, \textbf{input}: (\textit{Assembled\_Data.csv}, \textit{Analyzed Course}, $ \apv $, Group Segmentation $ \big( I_{\ell}: \ell \in [L] \big) $)\;
	\textbf{call} Algorithm \ref{Alg Optimization Algoritm}, \textbf{input}:  (\textit{Assembled\_Data.csv},
	\textit{Year}, \textit{Semester}, \textit{Analyzed Course}, $ \apv $, Group Segmentation $ \big( I_{\ell}: \ell \in [L] \big) $, $ \mt $ )\;
	
	APV\_mt\_Assessment[Year, Semester]$ \leftarrow \rho_{\mt} $. 
	
	}
}
\caption{Historical Assessment Algorithm}\label{Alg Analytica Omega}
\end{algorithm}
\begin{figure}[h!]
        \centering
        \begin{subfigure}[Example DC. Enhancement Results $ \apv = $ \textit{Pass Rate}. ]
{\includegraphics[scale = 0.38]{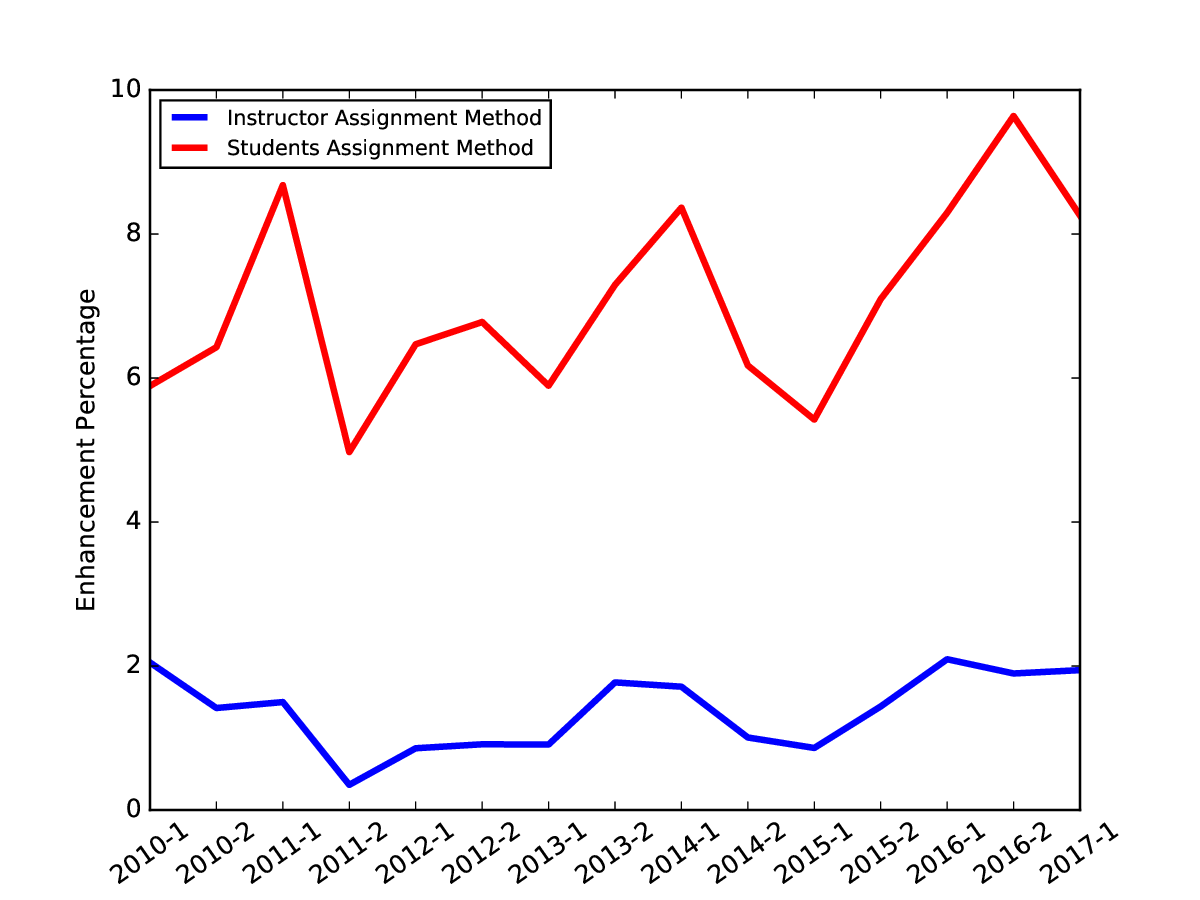} } 
        \end{subfigure}
        ~ 
          \begin{subfigure}[Example DC. Enhancement Results $ \apv = $ \textit{Average}.]
{\includegraphics[scale = 0.38]{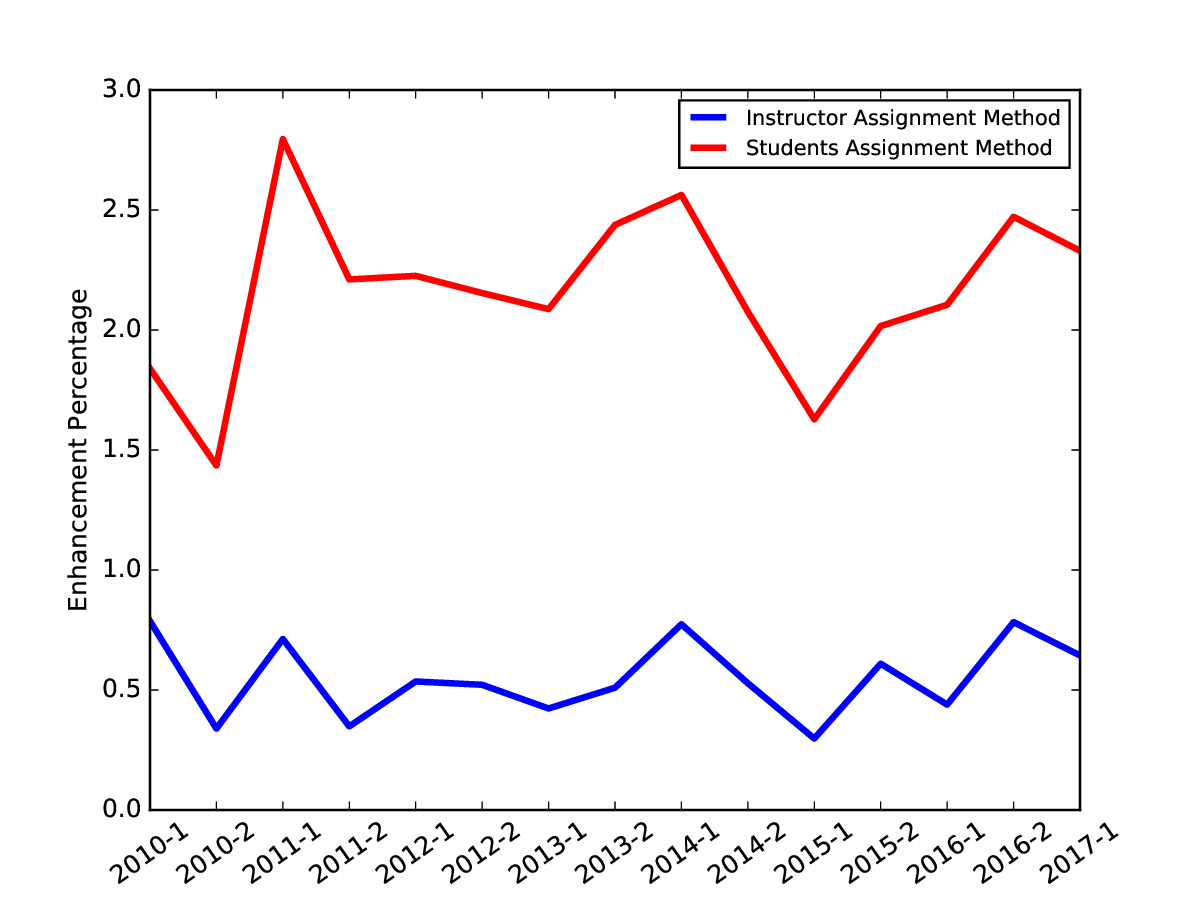} } 
        \end{subfigure}
\caption{Example: \textbf{Differential Calculus} course. Both figures show the enhancement results $ \rho_{\mt} $ for $ \mt \in \{ \ia, \sa\} $ optimization methods. The  Instructor Assignment method (IA) is depicted in blue, while the Students Assignment method (SA) is represented in red. 
Figure (a) shows results for the \textit{Pass/Fail} variable. 
Figure (b) shows the results for \textit{Average} variable. 
\label{Fig Historical Enhancement} }
\end{figure}
%
%
\begin{remark}[Economy Perspective: Figure \ref{Fig Historical Enhancement}]
 As it was already mentioned in the beginning of subsection \ref{Sec Historical Assessment}, the Students Assignment method ($ \sa $) yields better results than the Instructor Assignment method ($ \ia $). This is aligned with the following idea from the economic empirical perceptions: when individuals have more instruments to participate, their well-being in terms of social welfare increases.
\end{remark} 
%
%
%
%
\begin{table}[h!]
\def\arraystretch{1.4}
\small{
\begin{center}
\rowcolors{2}{gray!25}{white}
\begin{tabular}{p{1.5cm} | p{1.5cm} p{1.5cm} | p{1.5cm} p{1.5cm}  }
    \hline
    \rowcolor{gray!80}
Academic
&
\multicolumn{2}{c |}{\textit{APV} = Pass/Fail} &
\multicolumn{2}{c}{\textit{APV} = Average} \\
\rowcolor{gray!80}
Semester
& $ \mt = \ia $
& $ \mt = \sa $
& $ \mt = \ia $
& $ \mt = \sa $
 \\
\hline
2010-1	& 2.0482 &	5.8891  &  0.7868 &	1.8393 \\
2013-1	& 0.9090 &	5.8924	& 0.4230 & 2.0870 \\
2016-1	& 2.0939 &	8.2952	& 0.4391 &	2.1050 \\
\hline
\rowcolor{gray!80}
Mean & 1.3811 &	7.0432	& 0.5501 &	2.1584 \\
\hline
\end{tabular}
\end{center}
}
\caption{Relative Enhancements Sample, $ \rho_{\mt} $, 
	Course: \textbf{Differential Calculus}. We display the mean relative enhancement for the Differential Calculus course during the period from 2010-1 to 2017-1; together with a sample of three semesters in the same time window.}\label{Tb Historical Enhnacements}
\end{table}
%
%
%
%
%
%
%
%
\section{Randomization and Predictive Assessment}\label{Sec Randomization and Asymptotic Assesment}
%
%
%
%
%
%
So far, the method has been assessed with respect to the historical log i.e., comparing its optimization outputs with those of 15 recorded semesters. The aim of the present section is to perform Monte Carlo simulations on the efficiency of the method and apply the Law of Large Numbers to estimate the expected enhancement of the algorithm. We present below for the sake of completeness, its proof and details can be found in \cite{BillingsleyProb}.
\begin{theorem}[Law of Large Numbers]\label{Th the Law of Large Numbers}
Let $ \big(\Z^{(n)}:n\in \N\big) $ be a sequence of independent, identically distributed random variables with expectation $ \mu = \Exp(\Z^{(1)}) $, then
\begin{equation}\label{Eq the Law of Large Numbers}
\prob\bigg[ \Big\vert \frac{\Z^{(1)} + \Z^{(2)} + \ldots \Z^{(n)}}{n} - \mu \Big\vert > 0 \bigg]
\xrightarrow[n\, \rightarrow \,\infty]{} 0 ,
\end{equation}
i.e. , the sequence $ \big(\Z^{(n)}:n\in \N\big) $  converges to $ \mu $ in the Ces\`aro sense.
\end{theorem}
In order to achieve Monte Carlo simulations, we first randomize several factors/variables which define the setting of a semester for each course, in Section \ref{Sec Randomization of variables}. Next we discuss normalization criteria in Section \ref{Sec Normalization of the method}, to make the enhancement simulations comparable. Finally, we present in Section \ref{Sec Numerical Simulations}, the Monte Carlo simulations results for both, the random variable simulating the benefits of the method ($ \Z^{(n)} $ in Theorem \ref{Th the Law of Large Numbers}) as well as the evolution of its Ces\'aro means ($ \frac{\Z^{(1)} + \Z^{(2)} + \ldots + \Z^{(n)}}{n} $ in Theorem \ref{Th the Law of Large Numbers}) to determine the asymptotic performance of the proposed algorithm.   

Throughout this section we adopt a notational convention, the label \textbf{RandInputAlgorithm} will refer to the random versions of the respective algorithm developed in the previous sections. For instance
RandInputAlgorithm \ref{Alg Optimization Algoritm}, \textbf{input}:  
(\underline{Group Assignment Matrix  \textit{G} , List of Lecturers \textit{L\_list} }, \textit{Analyzed Course}, $ \apv $, Group Segmentation $ \big( I_{\ell}: \ell \in [L] \big) $, $ \mt $ ), refers to Algorithm \ref{Alg Optimization Algoritm} above, but with a different set of input data; for clarity the randomly generated input data are underlined. This notation is introduced for exposition brevity: avoiding to write an algorithm whose logic is basically identical to its deterministic version. 
%
%
%
%
\subsection{Randomization of Variables}\label{Sec Randomization of variables}
%
%
%
%
Four factors will be randomized in the same fashion: Number of Tenured Lecturers, Number of Enrolled Students, List of Students' GPA and Number of Groups. First, we randomize the integer-valued statistical variables by merely computing  95 percent confidence intervals from the empirical data and then assuming that, the impact of the factor can be modeled by a random variable uniformly distributed on such confidence interval, see \cite{MendenhallBeaver}.
\begin{definition}\label{Def Confidence Interval}
Let $ x $ be a scalar statistical variable with mean $ \bar{x} $, standard deviation $ \sigma $ and $ n $ its sample size.
\begin{enumerate}[(i)]
\item If $ x $ is real-valued, its \textbf{95 percent confidence interval} is given by 
\begin{equation}\label{Eq Confidence Interval Real-Valued}
 I_{x} \defining \Big[\bar{x} - 1.96 \, \frac{\sigma}{\sqrt{n}}, \bar{x} + 1.96 \, \frac{\sigma}{\sqrt{n}} \Big] .
\end{equation}
\item 
If $ x $ is integer-valued, its \textbf{95 percent confidence interval} is given by 
\begin{equation}\label{Eq Confidence Interval Integer-Valued}
 I_{x} \defining \Big[\big\lfloor\bar{x} - 1.96 \, \frac{\sigma}{\sqrt{n}} \big\rfloor,
 \big\lceil \bar{x} + 1.96 \, \frac{\sigma}{\sqrt{n}} \big\rceil\Big] \cap \setZ,
\end{equation}
where $ x $, $ \lfloor \cdot \rfloor ,\lceil \cdot \rceil: \R \rightarrow \R $ denote the floor and ceiling functions respectively. 
\end{enumerate}
\end{definition}
The randomization of the statistical variables listed above, heavily relies on the empirical distributions mined from the database.
\begin{hypothesis}\label{Hyp Variables Randomization}
\begin{enumerate}[(i)]
\item 
Let $ x $ be a scalar statistical variable, then its associated random variable $ \X $ is uniformly distributed on its confidence interval $ I_{x} $, i.e. $ \X \sim \unif (I_{x}) $, where the confidence interval is defined by \eqref{Eq Confidence Interval Integer-Valued} or \eqref{Eq Confidence Interval Real-Valued} depending on whether the variable $ x $ is integer or real valued.

\item Let $ \mathbf{x} = \big(x_{i},  \ldots, x_{d} \big)\in \R^{d} $ be a vector statistical variable, then its associated random variable is given by $ \textbf{X} = \big(\X_{1} , \ldots, \X_{d} \big) $, where $ \X_{i}  $ is the random variable associated to $ x_{i} $ for all $ i \in [d] $ as defined above.
\end{enumerate}
\end{hypothesis} 
From here, it is not hard to compute the confidence intervals (or ranges) of the random variables as it is shown in \textsc{Tables} \ref{Tb Tenured Random Variable} and \ref{Tb Enrollement Random Variable}. In contrast, the \textit{Sections} and the \textit{GPA} variables will need further considerations in its treatment. 

\begin{table}[h!]
\def\arraystretch{1.4}
\small{
\begin{center}
\rowcolors{2}{gray!25}{white}
\begin{tabular}{ c | c c c c c c c c }
    \hline
    \rowcolor{gray!80}
\diagbox 
{PARAMETERS}{COURSE}
& DC
& IC
& VC
& VAG
& LA
& ODE
& BM
& NM
 \\
\hline
Upper Extreme &
8 &	5 &	3 & 7 &	5 &	4 &	4 &	2 \\
Lower Extreme &
6 &	3 &	2 &	4 &	3 &	3 &	2 &	1 \\
Average &
7.2667 & 4.0667	& 2.6000 &	5.1333 & 3.7333	& 3.5333 & 3.3333 &	1.6667 \\
Standard Deviation &
0.7037 &	1.1629	& 0.7368 & 1.9952 &	1.2228 & 0.9155	& 1.0465 & 0.6172 \\
\hline
\end{tabular}
\end{center}
}
\caption{Random Variable: Number of Tenured Instructors $ \upnt $, 
	Course: \textbf{All}.  The upper \& lower extremes, average and standard deviation for the random variable ``Number of Tenure Instructors NT'' across all courses are displayed.}\label{Tb Tenured Random Variable}
\end{table}
\begin{table}[h!]
\def\arraystretch{1.4}
\scriptsize{
\begin{center}
\rowcolors{2}{gray!25}{white}
\begin{tabular}{ c | c c c c c c c c }
    \hline
    \rowcolor{gray!80}
\diagbox 
{PARAMETERS}{COURSE}
& DC
& IC
& VC
& VAG
& LA
& ODE
& BM
& NM
 \\
\hline
Upper Extreme &
1554 &	1203 &	586	& 1243	& 1045	& 882 & 974	& 301 \\
Lower Extreme &
1337 &	1043 & 513 & 1050 &	932	& 721 &	837 & 234 \\
Average &
1445.9333 &	1122.8667 &	549.8000 & 1146.5333 &	988.6667 &	801.8000 &	905.4000 &	267.5333 \\
Standard Deviation &
213.4446 &	156.8642 &	70.9969	& 188.8304 & 110.7229 &	158.4433 &	134.2971 &	65.4847 \\
\hline
\end{tabular}
\end{center}
}
\caption{Random Variable: Number of Enrolled Students $ \upne $, 
	Course: \textbf{All}. 
The upper \& lower extremes, average and standard deviation for the random variable ``Number of Tenure Students NE'' across all courses are displayed.
}\label{Tb Enrollement Random Variable}
\end{table}

The \textit{Sections} variable is a list of several sections with different capacities. A statistical scan of the data shows that this list is a most unpredictable variable, because section capacities range from 15 to 150 with very low relative frequencies in each of its values. Consequently, it was decided to group the section capacities in integer intervals
\begin{definition}\label{Def Gathering of Section Capacities}
Given the list of integer intervals 
\begin{equation}\label{Eq Intervals Agrupation} 
\mathcal{I} \defining \big\{ [15, 30], [31, 45], [46, 60], [61, 75], [76, 90], 
[91, 105], [ 106, 120], [121, 135], [136, 150]\big\},  
\end{equation}
%
%
%
%
for each semester and for each course, the \textbf{sections frequency variable} is given by 
\begin{equation}\label{Eq Sections Frequency Variable}
\boldsf \defining  
\Big(\frac{\itns_{I}}{\sum\limits_{K\, \in \,\mathcal{I} } \itns_{K} } : I \in \mathcal{I} \Big) ,
\end{equation}
where $ \itns_{I} $ is the number of sections whose capacity belongs to the interval $ I \in \mathcal{I} $.
%
\end{definition}
\begin{hypothesis}\label{Hyp Group Sizes Variable}
The Sections variable $ \textbf{S} $ is completely defined by the number of groups variable $ \upns $ in the following way
\begin{equation}\label{Eq Group Sizes Variable}
\textbf{S} \defining \Big\lceil \upns \, \overline{\boldsf} \Big\rceil .
\end{equation}
Here $ \overline{\boldsf} $ is the average vector of $ \boldsf $ introduced in \textsc{Equation} \eqref{Eq Sections Frequency Variable} and it is understood that the ceiling function $ \lceil \cdot \rceil $ applies to each component of the vector. 
\end{hypothesis}
Finally, the \textit{GPA} variable is treated as follows
%
%
%
%
\begin{hypothesis}\label{Hyp GPA random variable}
For each semester define $\mathbf{x} \defining \big(x_{i}: i\in [50] \big)$, where $ x_{i} $ is the relative frequency of registering students whose GPA is equal to $ \frac{i}{10} $; in particular $ \sum_{i \, \in\, [50] } x_{i} = 1 $. Let $ \textbf{X}_{\gpa} $ be the associated random variable to the list of relative frequencies $ \mathbf{x} $, as introduced in \textsc{Hypothesis} \ref{Hyp Variables Randomization}. Then, the random variable $ \textbf{GPA} $ is given by
\begin{equation}\label{Eq GPA random variable}
\textbf{GPA}\defining \big\lceil \upne \, \textbf{X}_{\gpa} \big\rceil ,
\end{equation}
where $ \upne $ is the number of enrolled students random variable and it is understood that the ceiling function $ \lceil \cdot \rceil $ applies to each component of the vector. 
\end{hypothesis} 
\begin{remark}\label{Rem Sections and GPA rand var}

Notice that both random variables $ \textbf{S} $ and $ \textbf{GPA} $ are the product of a scalar and a vector. However, for $ \textbf{S} $ the scalar is a random variable $ \upns $ and the vector $ \overline{\boldsf} $ is deterministic, while in the case of $ \textbf{GPA} $ the scalar $ \upne $ and the vector $ \textbf{X}_{\text{GPA}} $ are random variables. There lies the difference in the randomization of the vector variables.

\end{remark}
So far, $ \textbf{S} = \big( s_{K}^{(1)}: K\in \mathcal{I}\big) $ is producing a list of sections whose capacity lies within the ranges declared in $ \mathcal{I} $ (\textsc{Equation} \eqref{Eq Intervals Agrupation}), this introduces a set of slacks which will be used later on, to match the number of enrolled students $ \upne $ with the total sections capacity. The match will be done in several steps, once the equality $ \sum_{K\, \in \, \mathcal{I} }s_{K} = \upne $ is attained, a group assignment matrix $ G $ (as in Definition \ref{Def Performance and Group Assignment Matrices} (i)) will be generated randomly.
\begin{enumerate}[{step} 1.]
\item\label{Step 1 Initial Approximation Solution} Solve the following Data Fitting Problem (see \cite{Bertsimas} for its solution)
\begin{problem}\label{Pblm Optimization Students Sections Capacity}
Given two realizations of $ \textbf{S} = \big( s_{K}^{(1)}: K\in \mathcal{I}\big) $ and $ \upne $, consider the integer problem
\begin{equation}\label{Eq Optimization Students Sections Capacity}
df^{(1)} = \min \bigg\{\Big\vert \sum_{K\, \in \, \mathcal{I}}\sum_{i\, = \, 1}^{s_{K}^{(1)}} x_{K, i}
- \upne \Big\vert :  
x_{K, i}\, \in K , \text{for all } i \in \big[ s_{K}^{(1)} \big] \text{ and }  K \in \mathcal{I} \bigg\}.
\end{equation}
Denote by $ \big(x_{K, i}^{(1)}: i \in [s_{K}^{(1)}], K \in \mathcal{I}\big) $ the optimal solution to problem \eqref{Eq Optimization Students Sections Capacity}. If $ df^{(1)} \equiv 0 $ then jump to \textsc{step} \ref{Step 4 Initial Generate randomly a group assignment matrix} below.
%
\end{problem}
%
%

\item\label{Step 2 Increase/Decrease Sections} Decide whether or not is more convenient increase or decrease the number of sections $ s_{K}^{(1)} \mapsto s_{K}^{(2)} $ according the case, using the \textbf{Greedy Algorithm} \ref{Alg Greedy Algorithm}, to modify the sections' capacity, starting from the large sections to the small ones and get
\begin{equation}\label{Eq Opening or closineg sections}
df^{(2)} \defining \Big\vert \sum_{K\, \in \, \mathcal{I}}\sum_{i\, = \, 1}^{s_{K}^{(2)}} 
x_{K, i}^{(1)}
- \upne \Big\vert < 
df^{(1)} 
.
\end{equation}
\begin{algorithm}[H]
\KwData{
$ \big(s_{K}^{(1)}: K \in \mathcal{I} \big) $ ,
$  \Sigma = \sum_{K\, \in \, \mathcal{I}}\sum_{i\, = \, 1}^{s_{K}^{(1)}} x_{K, i}^{(1)} $,
$ \upne $. }
\KwResult{New set of sections quantities $ \big(s_{K}^{(2)}: K \in \mathcal{I} \big) $.
}
\textbf{Initialization}\;
\textbf{sort} the list of intervals $ \mathcal{I} $ from large values to small ones\;
	\eIf{ $ \Sigma > \upne $ }{
	\textbf{define} $ df^{(2)} \defining \Sigma - \upne $, $ K^{(r)} \defining \text{ right extreme of } K $, for all $ K \in \mathcal{I} $\;
	\While{ $ df^{(2)} > \min\{K^{(r)}: K \in \mathcal{I} \} $ }{
		\If{ $ \Sigma - \upne > K^{(r)}  $  }{
		$ s_{K}^{(2)} = s_{K}^{(1)} - 1 $, $ df^{(2)} = df^{(2)} - K^{(r)} $
		}
		}
		}
	{
	\textbf{define} $ df^{(2)} \defining \upne - \Sigma $, $ K^{(\ell)} \defining \text{ left extreme of } K $, for all $ K \in \mathcal{I} $\;
	\While{ $ df^{(2)} > \min\{K^{(\ell)}: K \in \mathcal{I} \} $ }{
	\If{ $ \Sigma - \upne > K^{(\ell)}  $  }{
	$ s_{K}^{(2)} = s_{K}^{(1)} + 1 $, $ df^{(2)} = df^{(1)} - K^{(\ell)} $
		}
		}
	}
\caption{Greedy Algorithm Increase/Decrease Number of Sections}\label{Alg Greedy Algorithm}
\end{algorithm}
If $ df^{(2)} \equiv 0 $ then jump to \textsc{step} \ref{Step 4 Initial Generate randomly a group assignment matrix} below.

\item\label{Step 3 Modify Sections Capacity} If $ 0 < df^{(2)} $ (once the Greedy Algorithm \ref{Alg Greedy Algorithm} has been applied), apply \textbf{increase/decrease capacities Algorithm} \ref{Alg Greedy Algorithm Sections Capacity} (breaking the constraints $ x_{K, i}^{(2)} \in K $ of \textsc{Equation} \eqref{Eq Optimization Students Sections Capacity}). Firstly changing $ x_{K, i}^{(1)} \mapsto x_{K, i}^{(2)} $ as evenly as possible. Secondly, distributing the reminder in randomly chosen sections $ x_{K, i}^{(2)} \mapsto x_{K, i}^{(3)} $ and get
\begin{equation}\label{Eq Evenly Distribute Excess or Absence}
df^{(3)} \defining \Big\vert \sum_{K\, \in \, \mathcal{I}}\sum_{i\, = \, 1}^{s_{K}^{(2)}} 
x_{K, i}^{(3)}
- \upne \Big\vert \equiv 0 < df^{(2)}.
\end{equation}
\begin{algorithm}[H]
\KwData{
$ df^{(2)} $, 
$ \big(s_{K}^{(2)}: K \in \mathcal{I} \big) $ ,
$  \Sigma = \sum_{K\, \in \, \mathcal{I}}\sum_{i\, = \, 1}^{s_{K}^{(2)}} x_{K, i}^{(1)} $,
$ \upne $. }
\KwResult{New set of sections capacities $ \big(x_{K, i}^{(2)}: i \in [ s_{K}^{(2)} ], K \in \mathcal{I} \big) $.
}
\textbf{initialization}\;
\textbf{sort} the list of sections capacities $ \big( x_{K, i}^{(1)}: i \in [s_{K}^{(2)}], K \in \mathcal{I} \big) $ from large values to small ones\;
\textbf{define} total number of sections $ s \defining \sum_{K \in \mathcal{I} } s_{K}^{(2)} $\; 
	\eIf{ $ \Sigma > \upne $ }{
	\textbf{define} $ u = \big\lfloor  \dfrac{\Sigma - \upne }{s} \big\rfloor $ \;
	\textbf{define} $ x_{K, i}^{(2)} \defining x_{K, i}^{(1)} - u $ for all $ i \in [ s_{K}^{(2)} ], K \in \mathcal{I} $\;
	\textbf{choose} $ \Sigma - \upne - u \times s $ sections $ K \in \mathcal{I}, i \in [s_{K}^{2} ] $ denote this set by $ S $\;
	\eIf{ $ (K, i) \in S $ }{
	$ x_{K, i}^{(3)} \defining x_{K, i}^{(2)} - 1 $
	}
	{
	$ x_{K, i}^{(3)} \defining x_{K, i}^{(2)} $
	}
	}
	{
	\textbf{define} $ u = \big\lfloor  \dfrac{\upne - \Sigma }{s} \big\rfloor $ \;
	\textbf{define} $ x_{K, i}^{(2)} \defining x_{K, i}^{(1)} + u $ for all $ i \in [ s_{K}^{(2)} ], K \in \mathcal{I} $\;
	\textbf{choose} $ \upne - \Sigma - u \times s $ sections $ K \in \mathcal{I}, i \in [s_{K}^{2} ] $ denote this set by $ S $\;
	\eIf{ $ (K, i) \in S $ }{
	$ x_{K, i}^{(3)} \defining x_{K, i}^{(2)} + 1 $
	}
	{
	$ x_{K, i}^{(3)} \defining x_{K, i}^{(2)} $
	}
	}
\caption{Greedy Algorithm Increase/Decrease Sections' Capacity}
\label{Alg Greedy Algorithm Sections Capacity}
\end{algorithm}

\item\label{Step 4 Initial Generate randomly a group assignment matrix} Generate randomly a group assignment matrix $ G $.
\end{enumerate}
The random setting described above is summarized in the pseudocode \ref{Alg Random Setting MC Analytica Omega}

\begin{algorithm}[t]
\KwData{ 
$ \upne $ random variable distribution, 
$ \textbf{X}_{\gpa} $ random variable distribution \\
$ \upns $ random variable distribution, average sections frequency $ \overline{\boldsf} $ variable \\
Analyzed Course: DC, IC, ..., NM. \\
List of Tenured Lecturers}
\KwResult{Random Group Assignment Matrix $ G  $.\\
}
\textbf{Initialization}\;
\textbf{compute} a realization of $ \upne $ and a realization for $ \textbf{X}_{\gpa} $\;
\textbf{compute} $ \textbf{GPA}, $ \textbf{input}: ($  \upne, \textbf{X}_{\gpa} $)\;
\textbf{call} RandInputAlgorithm \ref{Alg Segmentation of Students}, \textbf{input}: 
(\underline{\textbf{GPA} list})\;
\textbf{compute} a realization of $ \upns $ \;
\textbf{compute} $ \textbf{S} $, \textbf{input}: ($  \upns, { \overline{\boldsf} }$, \textit{Analyzed Course})\;
\textbf{solve} Problem \ref{Eq Optimization Students Sections Capacity} \textbf{input}: ($ \textbf{S}, \upne $)\;
\eIf{ $ df^{(1)} > 0 $ }{
	\textbf{run} the increase/decrease number of sections Greedy Algorithm \ref{Alg Greedy Algorithm}\;
	\eIf{ $ df^{(2)} > 0 $ }{
		\textbf{run} the increase/decrease capacities algorithm 
		\ref{Alg Greedy Algorithm Sections Capacity} \; 
		$ \textbf{S} \leftarrow \big(x_{K, i}^{(3)}: i\in [s_{K}^{(2)}], K \in \mathcal{I} \big) $\;
		\textbf{return} $ \textbf{S} $
		}
		{$ \textbf{S} \leftarrow \big(x_{K, i}^{(3)}: i\in [s_{K}^{(2)}], K \in \mathcal{I} \big) $\;
		\textbf{return} $ \textbf{S} $}
	}
	{$ \textbf{S} \leftarrow \big(x_{K, i}^{(3)}: i\in [s_{K}^{(1)}], K \in \mathcal{I} \big) $\;
	\textbf{return} $ \textbf{S} $ }

\textbf{compute} a random group matrix assignment $ G $, \textbf{input}: $ (\textbf{S},\textbf{GPA} ) $\;
\caption{Random Setting Algorithm}
\label{Alg Random Setting MC Analytica Omega}
\end{algorithm}
We close this section displaying tree tables. \textsc{Table} \ref{Tb Sections Random Variable} contains the confidence intervals for the number of sections random variable $ \upns $.  The capacities distribution vector $ \overline{\boldsf}  $ , as well as the confidence intervals of $ \textbf{X}_{\gpa} $ are shown in \textsc{Table} \ref{Tb GPA Random Variable} for the course of Differential Calculus only; due to the $ \gpa $ range length, the table has been split in five rows to fit the page format. Finally, \textsc{Table} \ref{Tb Sections Realization Example} presents an example of a group assignment matrix $ G $ produced by Algorithm \ref{Alg Random Setting MC Analytica Omega}
\begin{table}[h!]
\def\arraystretch{1.4}
\small{
\begin{center}
\rowcolors{2}{gray!25}{white}
\begin{tabular}{ c | c c c c c c c c }
    \hline
    \rowcolor{gray!80}
\diagbox 
{PARAMETERS}{COURSE}
& DC
& IC
& VC
& VAG
& LA
& ODE
& BM
& NM
 \\
\hline
Upper Extreme &
20	& 11 &	5 &	17 & 10	& 7	& 17 &	3 \\
Lower Extreme &
16	& 9	& 3	& 15 &	8 &	5 &	13 & 1 \\
$ [15, 30] $  &
0.0202  & 0.0000 &	0.0133  & 0.0000	& 0.0000	& 0.0000	& 0.0064	& 0.0000 \\
$ [31, 45] $ &
0.0030  & 0.0000 &	0.0000	& 0.0000	& 0.0000	& 0.0000	& 0.0712	& 0.0000 \\
$ [46, 60] $ &
0.0403	& 0.0310 &	0.0000	& 0.0044	& 0.0000	& 0.0000	& 0.0609	& 0.0000 \\
$ [61, 75] $ &
0.3802	& 0.0330 &	0.0000	& 0.7112	& 0.0572	& 0.0222	& 0.6125	& 0.0667 \\
$ [76, 90] $ &
0.1155	& 0.0048 &	0.0000	& 0.1134	& 0.0083	& 0.0000	& 0.2056	& 0.0000 \\
$ [91, 105] $ &
0.1034	& 0.0588 &	0.0133	& 0.0675	& 0.0763	& 0.0095	& 0.0300	& 0.0333 \\
$ [106, 120] $ &
0.1640	& 0.2115 &	0.0300	& 0.1035	& 0.2939	& 0.0429	& 0.0133	& 0.0333 \\
$ [121, 135] $ &
0.0629	& 0.3417 &	0.2600	& 0.0000	& 0.1791	& 0.3937	& 0.0000	& 0.2889 \\
$ [135, 150] $ &
0.1104	& 0.3193 &	0.6833	& 0.0000	& 0.3852	& 0.5317	& 0.0000	& 0.5778 \\
\hline
\end{tabular}
\end{center}
}
\caption{Random Variable Number of Sections $ \upns $ and Capacities Distribution Vector 
	$ \overline{\boldsf} $, 
	Course: \textbf{All}. 
The upper \& lower extreme and confidence intervals for the Variable number of sections $ \upns $ are displayed  for all courses.}\label{Tb Sections Random Variable}
\end{table}
\begin{table}[h!]
\def\arraystretch{1.4}
\scriptsize{
\begin{center}
\begin{tabular}{ c | c c c c c c c c c c}
    \hline
    \rowcolor{gray!80}
\diagbox 
{INTERVALS}{GPA}
& 0.1
& 0.2
& 0.3
& 0.4
& 0.5
& 0.6
& 0.7
& 0.8
& 0.9
& 1.0
 \\
\hline
Upper Extreme &
0.0021 &	0.0027 & 0.0031 & 0.0031 &	0.0035	& 0.0031	& 0.0044	& 0.0040 &	0.0047 & 0.0046\\
\rowcolor{gray!25}
Lower Extreme &
0.0012& 0.0013 &	0.0017	& 0.0017 &	0.0024	& 0.0016	& 0.0031	& 0.0021 &  0.0020 & 0.0028 \\ 
\hline
\rowcolor{gray!80}
\diagbox 
{INTERVALS}{GPA}
& 1.1
& 1.2
& 1.3
& 1.4
& 1.5
& 1.6
& 1.7
& 1.8
& 1.9
& 2.0
 \\
\hline
Upper Extreme & 
0.0047	& 0.0045	& 0.0055	& 0.0053	& 0.0055	& 0.0056	& 0.0051	& 0.0073	& 0.0082 & 0.0084\\
\rowcolor{gray!25}
Lower Extreme &
0.0026	& 0.0025	& 0.0032	& 0.0028	& 0.0031	& 0.0038	& 0.0033	& 0.0049	& 0.0058 & 0.0056
\\
\hline
\rowcolor{gray!80}
\diagbox
{INTERVALS}{GPA}
& 2.1
& 2.2
& 2.3
& 2.4
& 2.5
& 2.6
& 2.7
& 2.8
& 2.9
& 3.0
 \\
\hline
Upper Extreme &
0.0114	& 0.0119	& 0.0147	& 0.0191	& 0.0208	& 0.0277	& 0.0328	& 0.0348	& 0.0420 & 0.0858\\
\rowcolor{gray!25}
Lower Extreme & 
0.0077	& 0.0076	& 0.0105	& 0.0142	& 0.0165	& 0.0210	& 0.0250	& 0.0287	& 0.0332 & 0.0672
\\
\hline
\rowcolor{gray!80}
\diagbox 
{INTERVALS}{GPA}
& 3.1
& 3.2
& 3.3
& 3.4
& 3.5
& 3.6
& 3.7
& 3.8
& 3.9
& 4.0
 \\
\hline
Upper Extreme &
 0.0571	& 0.0566	& 0.0671	& 0.0668	& 0.0630	& 0.0668	& 0.0625	& 0.0527 & 0.0496 & 0.0421
\\
\rowcolor{gray!25}
Lower Extreme &
0.0506	& 0.0489	& 0.0565	& 0.0594	& 0.0558	& 0.0549 &	0.0526	& 0.0440	& 0.0404 & 0.0357
\\
\hline
\rowcolor{gray!80}
\diagbox
{INTERVALS}{GPA}
& 4.1
& 4.2
& 4.3
& 4.4
& 4.5
& 4.6
& 4.7
& 4.8
& 4.9
& 5.0
 \\
\hline
Upper Extreme &
0.0373	& 0.0255	& 0.0209	& 0.0161	& 0.0136	& 0.0081	& 0.0055	& 0.0029	& 0.0011 & 0.0002
\\
\rowcolor{gray!25}
Lower Extreme &
0.0282	& 0.0201	& 0.0160	& 0.0115	& 0.0088	& 0.0054	& 0.0032	& 0.0014	& 0.0002 & 0.0000
\\
\hline
\end{tabular}
\end{center}
}
\caption{Random Variable: $ \textbf{X}_{\gpa} $, 
	Course: \textbf{Differential Calculus}. 
	The distribution capacities vector $ \overline{\boldsf}  $ and the confidence intervals of $ \textbf{X}_{\gpa} $ for the course of Differential Calculus are displayed. }\label{Tb GPA Random Variable}
\end{table}
\subsection{Normalization of the method and probabilistic spaces}\label{Sec Normalization of the method}
In order to measure the proposed method's enhancement, now there is need to normalize the results as in the case of the historical assessment of the algorithm, \textsc{Section} \ref{Sec Optimization and Historical Assessment}, \textsc{Equation} \eqref{Eq Performance Rates} where the improvement in the academic performance variable was divided over the historical performance of the semester at hand. In the case of Monte Carlo simulations, the concept of ``historical performance" simply does not apply, as the assignation of students and/or lecturers actually did not happen. We approach this fact in two different ways
\begin{definition}[Normalization Methods]\label{Def Normalization Methods}
We introduce the following normalization methods.
\begin{enumerate}[(a)]
\item\label{Mthd Random Normalization} \textbf{Random Normalization}. Normalize with respect to a random assignation of instructors or students, depending on the method \textbf{IA} or \textbf{SA} respectively.

\item\label{Mthd Expected Normalization} \textbf{Expected Normalization}. Normalize with respect to the \textit{expected assignation} of instructors or students, depending on the method \textbf{IA} or \textbf{SA} respectively. 
\end{enumerate}
\end{definition}
In the first case, it is straightforward to compute the normalization, in the second case, the concept of \textit{expected assignation} needs to be stated in neater terms. To that end, we need to present some intermediate mathematical results and definitions
%
%
%
%
\begin{theorem}\label{Th expected performance molecule computation}
Let $ K \in \N $ be fixed and let $\big( T(i, k): i, k\in [K]\big) $ be a matrix. Define the Random Variable 
\begin{align*}
& \X_{ \ia }: \itS_{K} \rightarrow \R, & 
& \X_{ \ia }(\sigma) \defining \sum\limits_{k \, \in \,  [K]} T\big(k, \sigma(k)\big) .
\end{align*}
Then,
\begin{equation}\label{Eq expected performance molecule computation}
\Exp\big(\X_{ \ia }\big) = 
\frac{1}{K}\sum\limits_{i, k \, \in\, [K]} T(k, i)  = \frac{1}{K}\, \msum(T),
\end{equation}
where $ \msum(T) \defining \sum\limits_{(k,i) \in [K] \times [K]}T(k, i)  = \sum\limits_{k \,  \in \, [K] }\sum\limits_{i \, \in  [K]} T(k, i) $.
\end{theorem}
\begin{proof}
Consider the following calculation
\begin{equation*}
\begin{split}
\Exp\big(\X_{ \ia }\big) 
& = \frac{1}{K!}\sum\limits_{\sigma\,\in \, \itS_{K}} \X_{ \ia }(\sigma) 
= 
\frac{1}{K!}\sum\limits_{\sigma\,\in \, \itS_{K} } \sum\limits_{k \, \in\, [K]} T\big(k,\sigma(k)\big) 
= \frac{1}{K!} \sum\limits_{k \, \in\, [K]} \sum\limits_{\sigma\,\in \, \itS_{K}} T\big(k,\sigma(k)\big) 
\\
& = \frac{1}{K!} \sum\limits_{k \, \in\, [K]} \sum\limits_{i \, \in\, [K]}
\sum\limits_{\substack{\sigma\,\in \, \itS_{K} \\ \sigma(k) \, = \, i } } T\big(k,\sigma(k)\big) 
= \frac{1}{K!} \sum\limits_{k \, \in\, [K]} \sum\limits_{i \, \in\, [K]}
T\big(k, i\big)\sum\limits_{\substack{\sigma\,\in \, \itS_{K} \\ \sigma(k) \, = \, i } }  1
= \frac{(K - 1)!}{K!} \sum\limits_{k \, \in\, [K]} \sum\limits_{i \, \in\, [K]}
T\big(k, i\big) .
\end{split}
\end{equation*}
From here, Equation \eqref{Eq expected performance molecule computation} follows trivially.
\end{proof}
\begin{remark}\label{Rem expected performance molecule computation}
Notice that in Proposition \ref{Th expected performance molecule computation} the following holds
\begin{enumerate}[(i)]
\item It is understood that the probabilistic space is $ \Omega \equiv \itS_{K} $ where all the outcomes are equally likely.

\item Assume that the setting of a semester is given, namely: number of sections with capacities, conformation of sections and a set of instructors teaching the course. Then, $ T = C_{\apv} $, $ J $ is the number of sections and $ \Exp(X_{ \ia }) $ is the expected performance, when assigning instructors randomly to the available sections with defined students i.e., the $ \ia $ method.

\end{enumerate}
\end{remark}
Our next step is to be able to compute the expected performance of a group when assigning students randomly to available sections with defined instructors. This task is far more complicated due to the richness of the search space. We begin introducing some notation.
\begin{definition}\label{Def labeling and asignation functions}
Let $ N, L, J \in \N $ , $ \p = (p_{1}, \ldots, p_{L}) \in \N^{L}, \g = (g_{1}, \ldots, g_{J}) \in \N^{J} $ be as in Definition \ref{Def Characteristic Numbers}. 
\begin{enumerate}[(i)]
\item Let $ c: [N] \rightarrow [L] $ be the \textbf{classification function} of each student i.e., for each student $ n \in [N] $ it assigns the label $ c(n) \in [L] $ describing the profile to which he/she belongs to. 

\item Define the \textbf{student assignation} probabilistic space by
\begin{equation}\label{Eq student assignation probabilitic space}
\Omega\defining \Big\{\omega: [N] \rightarrow [J]: 
\big\vert \omega^{-1}(j) \big\vert = g_{j} \, , \, \text{for all } j \in [J]\Big\} .
\end{equation}

\item For a fixed element $ \omega \in \Omega $, define the matrix $ G^{\omega} \in \R^{L\times J} $ whose entries are given by 
\begin{align*}
& G^{\omega}(\ell, j) =  
\big\vert \big\{n\in [N]: c(n) = \ell, \omega(n) = j  \big\}\big\vert = 
\big\vert c^{-1}(\ell) \cap \omega^{-1}(j)\big\vert\, ,
&
& \forall\, \ell \in [L], j\in [J] .
\end{align*}
\end{enumerate}
\end{definition}
\begin{remark}\label{Rem Student Assignment Setting}
In Definition \ref{Def labeling and asignation functions} notice the following
\begin{enumerate}[(i)]
\item The student classification function satisfies that $ p_{\ell} \defining \big \vert c^{-1}(\ell) \big\vert $ for all $ \ell \in [L] $.

\item An element $ \omega $ of the student assignation space is such that every individual is assigned to a section and every section is full (recall the sum condition \eqref{Eq Sum Condition}). 

\item In our study, a list of $ N $ enrolled students is completely characterized a classification function $ c: [N] \rightarrow [L] $ and a section assignment function $ \omega \in \Omega $
\begin{equation*}
\begin{array}{cccc}
1, & 2, &   \ldots, & N , \\
c(1), & c(2), &  \ldots, & c(N) ,\\
\omega(1), & \omega(2), &  \ldots , & \omega(N) .
\end{array}
\end{equation*}
The first row represents identity and the second indicates profile classification. Therefore, only the third row is subject to decision or randomization as it is done in this model.

\item For every $ \omega \in \Omega $ the matrix $ G^{\omega} $ is clearly a group assignment matrix as introduced in Definition \ref{Def Performance and Group Assignment Matrices}.
\end{enumerate}
\end{remark}
%
%
\begin{proposition}\label{Th measuring global performance}
Let $ N, L, J \in \N $ , $ \p = (p_{1}, \ldots, p_{L}) \in \N^{L}, \g = (g_{1}, \ldots, g_{J}) \in \N^{J} $ be as in Definition \ref{Def Characteristic Numbers}. 
Then
\begin{align}\label{Eq measuring global performance}
& \trace\big(T_{\apv} G^{\omega} \big) = 
\sum\limits_{j\, = \, 1}^{J} \big(T_{\apv} G^{\omega} \big)\big(j, j \big) =
\sum_{n\, \in\,[N]} T_{\apv}\big( \omega(n), c(n) \big), 
& 
& \text{for all } \omega \in \Omega .
\end{align} 
\end{proposition}
\begin{proof}
Consider the following identities
\begin{equation*}
\begin{split}
\sum_{n\, \in\,[N]} T_{APV}\big( \omega(n), c(n) \big) & =
\sum_{(\ell, j)\, \in \,[L]\times[J]}
\sum_{\substack{n\, =\, 1\\ c(n) \, = \, \ell, \, \omega(n) \, = \, j}}^{N} T_{APV} 
\big(\omega(n), c(n)\big) \\
& = \sum_{j\, \in [J]} \sum_{\ell \, \in \,[L]} 
T_{APV} \big( j, \ell \big)
\sum_{\substack{n\, =\, 1\\ c(n) \, = \, \ell, \, \omega(n) \, = \, j}}^{N} 1 \\
& = \sum_{j\, \in [J]} \sum_{\ell \, \in \,[L]} 
T_{APV} \big( j, \ell \big)
G^{\omega}(\ell, j) 
\\
& = \sum_{j\, \in [J]} \big(T_{APV} G^{\omega}\big)  \big( j, j \big) ,
\end{split}
\end{equation*}
i.e., the result holds.
\end{proof}
\begin{remark}\label{Rem Probabilistic Modeling Student Assignation}
Observe that if it is assumed that that the instructor $t_{j}$ is assigned to the section $ j $ for all $ j \in {J}$, i.e., the instructor assignment function $ \pi \in \itS_{J} $ of Problem \ref{Pblm Students Assignment Problem}
is the identity then, the previous result states that 
\begin{equation}\label{Eq measuring global performance comment}
\trace\big(T_{\apv} G^{\omega} \big) = 
\sum\limits_{j\, = \, 1}^{J} \big(T_{\apv} G^{\omega} \big)\big(j, \pi(j) \big) =
\sum_{n\, \in\,[N]} T_{\apv}\big( \omega(n), c(n) \big), 
\end{equation} 
for each $ \omega \in \Omega $. Since the expression of the middle measures the global performance of the group, so does the right hand side. 
Hence, it makes sense to declare the left hand side in the expression above as a random variable.


\end{remark}
\begin{definition}\label{Def Random assignment matrix}
	Let $ N, L, J \in \N $ , $ \p = (p_{1}, \ldots, p_{L}) \in \N^{L}, \g = (g_{1}, \ldots, g_{J}) \in \N^{J} $ be as in Definition \ref{Def Characteristic Numbers} and let $ T \in \R^{ J\times L } = \big( T(j, \ell) :  k \in [J] , \ell \in [L] \big) $, be a fixed matrix.
	Define the \textbf{student assignment performance} random variable
	\begin{align}\label{Eq Student Assignment Performance}
	& \X_{ \sa }: \Omega \rightarrow \R\, ,  & 
	& \X_{ \sa }(\omega) \defining \sum\limits_{n\,\in\, [N] } T\big(\omega(n), c(n)\big) \, .
	\end{align}
\end{definition}
Before computing the expectation of the random variable $ \X_{\sa} $ some previous results from combinatorics are needed.
\begin{remark}[Definition \ref{Def Random assignment matrix}]
	The performance matrix $T$ in Definition \ref{Def Random assignment matrix} provides a measure for the performance of each instructor within each element of a given segmentation, as previously stated in Definition \ref{Def Performance and Group Assignment Matrices}. One possible performance matrix is given by the expected performance matrix, output of the Algorithm \ref{Alg Instructors' performance} in Section \ref{Sec Computation Lecturer Performance}. An ideal performance matrix $T$ should include more specific information about instructors as mentioned in Analysis \ref{measurament_of_instructor} (e.g., peer ratings, self-evaluation, alumni ratings, teaching awards and others).
\end{remark} 
\begin{lemma}\label{Th cardinal of student assignment space}
	\begin{enumerate}[(i)]
		\item The cardinal of the student assignment space is given by
		\begin{equation}\label{Eq cardinal of student assignment space}
		\vert \Omega \vert = N!\prod\limits_{j \, = \, 1}^{J}\dfrac{1}{g_{j}!} .
		\end{equation} 
		\item Let $ n\in [N] $, $ j\in [J] $ be fixed, and define the set
		\begin{equation*}
		\Omega_{n, j}  \defining \big\{ \omega\in \Omega: \omega (n) = j\big\}.
		\end{equation*}
		Then $ \vert \Omega_{n, j} \vert = \dfrac{(N - 1)! }{(g_{j} - 1)!} 
		\prod\limits_{\substack {  i \, \in \, [K]\\
				i \, \neq \, j} }\dfrac{1}{g_{i}!} $
	\end{enumerate}
\end{lemma}
\begin{proof}
	\begin{enumerate}[(i)]
		\item Let $ \omega $ be an element of $ \Omega $ and write it in the extended way i,e,
		\begin{equation*}
		\begin{array}{cccc}
				\omega(1), & \omega(2), & \ldots , & \omega(N), \\
			1 , & 2 ,  &  \ldots, & N .
		\end{array}
		\end{equation*}
		Clearly, $ \omega $ is a permutation of of the multiset
		\begin{equation}\label{Eq assignment section multiset}
		\big\{\underbrace{1,1,\ldots, 1}_{g_{1} \text{-times}},
		\underbrace{2,2,\ldots, 2}_{g_{2} \text{-times}},
		\ldots
		,\underbrace{J,J,\ldots, J}_{g_{J} \text{-times}}\big\} =
		\big\{1\cdot g_{1}, 2\cdot g_{2}, \ldots J\cdot g_{J} \big\}.
		\end{equation}
		%
	From elementary combinatorics, it is known that the number of permutations of the multiset \eqref{Eq assignment section multiset} is given given by the expression \eqref{Eq cardinal of student assignment space}, see Theorem 3.5 in \cite{BonaWalk}.
	
	\item First we analyze the case of the set $ \Omega_{N, j} $. Recalling the expression \eqref{Eq student assignation probabilitic space}, we can write $ \Omega_{N, j} = \big\{\omega: [N] \rightarrow [J]:  \, \omega(N) = j , \, 
	\vert \omega^{-1}(i) \vert = g_{i}, \text{ for all } i \in [J]\big\} $. It is direct to see that there is a bijection with the set $ \widetilde{\Omega} \defining \big\{\omega: [N - 1] \rightarrow [J]: 
	\big\vert \omega^{-1}(i) \big\vert = \widetilde{g}_{i} \, , \, \text{for all } i \in [J]\big\} $ where $ \widetilde{g}_{i} $ is defined as follows
	\begin{equation*}
	\widetilde{g}_{i} \defining 
	\begin{cases}
	g_{i}, & i \neq j ,  \\
	g_{i} - 1,  & j = i .
	\end{cases}
	\end{equation*}
	Applying the previous part on the set $ \widetilde{\Omega} $, it follows that $ \Omega_{N, j} $ satisfies the result. For the general case $ \Omega_{n,j} $, take the permutation $ \sigma \in \itS_{N} $ defined by 
	\begin{equation*}
	\sigma(k) \defining
	\begin{cases}
	N, & k = n , \\
	n, & k = N , \\
	k & \text{otherwise}. 
	\end{cases} 
	\end{equation*}
	Observe that the map $ \varphi: \Omega_{n,j} \rightarrow \Omega_{N,j} $ defined by $ \varphi(\omega) \defining \omega\circ\sigma $ is clearly a bijection. Consequently, $ \vert \Omega_{n,j} \vert = \vert \Omega_{N,j} \vert $ and the proof is complete.
	\end{enumerate}	
\end{proof}
\begin{theorem}\label{Def Expectation of random student assignment}
	The expectation of the random variable $ \X_{ \sa } $ is given by 
	\begin{equation}\label{Eq Expectation of random student assignment}
	\Exp\big(\X_{ \sa } \big) = \frac{1}{N} \,  \g^{t} T \, \p .
	\end{equation}
\end{theorem}
\begin{proof}
	By definition
	\begin{equation}
	\begin{split}
	\vert \Omega \vert\, \Exp\big(\X_{ \sa }\big)  
	& = 
	\sum\limits_{\omega\,\in\,\Omega} \sum\limits_{n\,= \,1 }^{N}
	T\big(\omega(n), c(n)\big)
	= \sum\limits_{n\,= \,1 }^{N}
	\sum\limits_{\omega\,\in\,\Omega}
	T\big(\omega(n), c(n)\big)\\
	& = \sum\limits_{n\,= \,1 }^{N}
	\sum\limits_{j\,= \,1 }^{J}
	\sum\limits_{\substack{\omega\,\in\,\Omega\\\omega(n) \, = \, j}}
	T\big(\omega(n), c(n)\big) 
	= \sum\limits_{n\,= \,1 }^{N}
	\sum\limits_{j\,= \,1 }^{J}
	T\big( j, c(n) \big)\sum\limits_{\substack{\omega\,\in\,\Omega\\\omega(n) \, = \, j}} 1 
	\end{split}
	\end{equation}
	Recalling Lemma \ref{Th cardinal of student assignment space} (ii), it follows that
	\begin{equation}
	\begin{split}
	\Exp\big(\X_{ \sa }\big)  
	& = \frac{1}{\vert \Omega \vert} 
	\sum\limits_{n\,= \,1 }^{N}
	\sum\limits_{j\,= \,1 }^{J}
	T\big( j, c(n) \big)
	\dfrac{(N - 1)! }{(g_{j} - 1)!} 
	\prod\limits_{\substack {  i \, \in \, [J]\\
			i \, \neq \, j} }\dfrac{1}{g_{i}!} \\
	& = \frac{1}{\vert \Omega \vert}\sum\limits_{n\,= \,1 }^{N}
	\sum\limits_{j\,= \,1 }^{J}
	g_{j} \, T\big( j, c(n) \big) (N - 1)! 
	\prod\limits_{  i \, \in \, [J] } \dfrac{1}{g_{i}!} = 
	\frac{1}{N}
	\sum\limits_{j\,= \,1 }^{J}
	\sum\limits_{n\,= \,1 }^{N}
	g_{j} \, T\big( j, c(n) \big) \\
	& = \frac{1}{N}
	\sum\limits_{j\,= \,1 }^{J}
	g_{j}
	\sum_{\ell\, = \, 1}^{L}
	\sum\limits_{\substack { n \, \in \, [N]\\
			c(n) \, = \, \ell} }
	\, T\big( j, c(n) \big)
	= \frac{1}{N}
	\sum\limits_{j\,= \,1 }^{J}
	g_{j}
	\sum_{\ell\, = \, 1}^{L}
	T\big(j, \ell\big) p_{\ell} \\
	& = \frac{1}{N}
	\sum_{\ell\, = \, 1}^{L}
	p_{\ell}
	\sum\limits_{j\,= \,1 }^{J}
	T\big(j, \ell \big) g_{j} .
	\end{split}
	\end{equation}
	Here, the second equality uses the identity $ \frac{1}{(g_{j} - 1)!} = \frac{ g_{j}}{g_{j}!} $ and the third uses the expression \eqref{Eq cardinal of student assignment space}, together with an obvious exchange of indexes. The fourth equality is a convenient association of summands, while the fifth merely uses the fact $ \vert c^{-1}(\ell) \vert = p_{\ell} $. From here, the result follows trivially.
\end{proof}
\begin{remark}\label{Rem random student assignment tool}
Let $ \pi \in \itS_{J} $ be a permutation and let $ A^{\pi} $ its associated permutation matrix
\begin{equation*}
A^{\pi} = \big[\eversor_{\pi(1)}, \eversor_{\pi(2)}, \ldots, \eversor_{\pi(J)} \big] ,
\end{equation*}
where $ \big(\eversor_{j}: j\in [J]\big) $ is the canonical basis of $ \R^{J} $. Then, if the instructors $ \big\{ t_{j}: j\in [J] \big\}$ are assigned to their corresponding sections by a permutation $ \pi \in \itS_{J} $, other than the identity, by taking
\begin{equation*}
T \defining T_{\apv} \, A^{\pi},
\end{equation*}
the random variable $ \X_{\sa}(\omega) $ (as defined in \eqref{Eq Student Assignment Performance}), computes the global performance of the group for each $ \omega \in \Omega $ (as discussed in Remark \ref{Rem Probabilistic Modeling Student Assignation}). Therefore, without loss of generality, it can be assumed that $ \pi  \in \itS_{J} $ is the identity.
\end{remark}
Finally we define 
\begin{definition}\label{Def Normalization Methods Optimization Random Version}
	The random version of Algorithm \ref{Alg Optimization Algoritm} will have two methods. 
	\begin{enumerate}[(i)]
		\item The Random Normalization method introduced in \textsc{Definition} \ref{Def Normalization Methods} \ref{Mthd Random Normalization} defined in Equation \ref{Eq Performance Rates}. However, it is important to observe that this time $ v_{\mt} $, $ \rho_{\mt} $ and $ \X_{\mt} \defining \sum_{j\, = \, 1}^{J} \big(T_{APV} G_{h}\big) \big(j, \pi_{h}(j) \big) $ are all random variables. 
		
		\item Second, the Expected Normalization method introduced in \textsc{Definition} \ref{Def Normalization Methods} \ref{Mthd Expected Normalization}, which is computed using 
		\begin{align}\label{Eq Performance Rates Normalized}
		& \gamma_{\mt} \defining 
		100\frac{v_{\mt} - \Exp\big(\X_{\mt}\big) }{ \Exp\big(\X_{\mt}\big) }
		\, , &
		& \mt \in \{ \ia, \sa\} .
		\end{align} 
		Here, $ \Exp\big(\X_{\mt}\big) $ is given by Theorem \ref{Th expected performance molecule computation} if $ \mt = \ia $ and by Theorem \ref{Def Expectation of random student assignment} if $ \mt = \sa $. Again, $ v_{\mt} $ and $ \gamma_{\mt} $ are both random variables. 
	\end{enumerate} 
\end{definition}
\begin{remark}\label{Rem Harmonic Mean Comments}
\begin{enumerate}[(i)]
\item It is understood that, for the application of the Law of Large Numbers \ref{Th the Law of Large Numbers} in the numerical experiments, the random variables above will be considered as sequences of independent, identically distributed, variables i.e., 
 $ \big(v_{\mt}^{(n)}: n\in \N \big) $, $ \big(\X_{\mt}^{(n)}: n\in \N\big) $, $ \big(\rho_{\mt}^{(n)}: n\in \N \big) $ and $ \big(\gamma_{\mt}^{(n)}: n\in \N \big) $; where the index $ n $ indicates an iteration of the Monte Carlo simulation. 
 
\item It is direct to see that $ \big(\gamma_{\mt}^{(n)}: n\in \N\big) $ converges in the Ces\`aro sense to $ \Exp\big(v_{\mt}^{(1)}\big) \big(\Exp\big(\X_{\mt}^{(n)}\big)  \big) ^{-1} - 1 $. 

\item Define $ \Z_{\mt}^{(n)} \defining \dfrac{1}{\X_{\mt}^{(n)} }$, since  $ \big(v_{\mt}^{(n)}: n\in \N \big) $ and $ \big(\X_{\mt}^{(n)}: n\in \N\big) $ are independent, it holds that 
\begin{equation}\label{Eq Harmonic Mean Comments}
\rho_{\mt}^{(n)} = \frac{v_{\mt}^{(n)} - \X_{\mt}^{(n)} }{  \X_{\mt}^{(n)} } 
= \frac{v_{\mt}^{(n)}  }{  \X_{\mt}^{(n)} } - 1
= v_{\mt}^{(n)}\Z_{\mt}^{(n)} - 1 
\xrightarrow[n\,\rightarrow \,\infty]{\text{Ces\`aro}} 
\Exp\big(v_{\mt}^{(1)}\big) \Exp\big(\Z_{\mt}^{(n)} \big) - 1
= \Exp\big(v_{\mt}^{(1)}\big) \Exp\Big(\frac{1}{\X_{\mt}^{(n)} } \Big) - 1.
\end{equation}
The right hand side of the expression above involves the reciprocal of the harmonic mean of the variable $ \big(\X_{\mt}^{(n)}: n\in \N\big) $. Clearly, $ \big(\gamma_{\mt}^{(n)}: n\in \N\big) $ and $ \big(\rho_{\mt}^{(n)}: n\in \N\big) $ converge (in the C\`esaro sense) to different limits. Unfortunately, the harmonic mean has no simple expression equivalent to that of Equation \eqref{Eq Expectation of random student assignment} for the arithmetic mean. Consequently, it can be handled only numerically; this will be done in the next section.  
\end{enumerate}
\end{remark}
%
%
%
%
\subsection{The Monte Carlo Simulation Algorithm and Numerical Results}\label{Sec Numerical Simulations}
The randomization of the variables as well as its normalization discussed in the sections \ref{Sec Randomization of variables} and \ref{Sec Normalization of the method} respectively are summarized in the pseudocode \ref{Alg Monte Carlo Analytica Omega} below. A particular example of the Monte Carlo simulation results is depicted in \textsc{Figure}  \ref{Fig Asymptotic Enhancement}, while the corresponding body/composition of enrolled students displayed presented in \textsc{Table} \ref{Tb Sections Realization Example}. 

The results of several simulations for the Differential Calculus course are summarized in \textsc{Table} \ref{Tb Monte Carlo Experiments Summary}. Out several experiments, it is observed that a reasonable level of convergence of the Ces\`aro means is attained above 800 iterations. Given that we are simulating the behavior of a highly complex random process, it is clear that no convergence rate can actually be concluded, the threshold for which the Ces\`aro mean stabilizes shifts significantly from one experiment to the other. This is because every experiment defines a number of sections $ \upns $, an enrollment body/composition of students as in \textsc{Table} \ref{Tb Sections Realization Example}, a group matrix assignment $ G $ and a number of tenured lecturers $ \upnt $, from here, the iteration process begins as it is shown in \textsc{Algorithm} \ref{Alg Monte Carlo Analytica Omega}. Therefore, the starting triple $ ( \upns, G, \upnt ) $ changes substantially between simulations as it can be seen in \textsc{Table} \ref{Tb Monte Carlo Experiments Summary}. These changes become even more dramatical when shifting from one course to another, as it is the case of \textsc{Table} \ref{Tb Monte Carlo Experiments All Courses Summary}, reporting the algorithm's performance for all the remaining seven service courses.

It is also important to observe the difference between Random (\textsc{Definition} \ref{Def Normalization Methods} \ref{Mthd Random Normalization}) vs. Expected (\textsc{Definition} \ref{Def Normalization Methods} \ref{Mthd Expected Normalization}) normalization methods. It is not significant in the simulation of the method's performance (see Figure \ref{Fig Asymptotic Enhancement} (a) and (b)) and  it is negligible in the behavior of their corresponding Ces\`aro means, i.e., regardless of the chosen normalization method (see figure Figure \ref{Fig Asymptotic Enhancement} (c) and (d)), the asymptotic behavior difference is negligible at least, from the numerical point of view. The latter can be also observed on the Tables  \ref{Tb Monte Carlo Experiments Summary} and \ref{Tb Monte Carlo Experiments All Courses Summary}.
\begin{remark}[Figure \ref{Fig Asymptotic Enhancement}]
	Figure \ref{Fig Asymptotic Enhancement} depicts enhancement (and Ces\`aro means enhancement) results for the variable Pass Rate of the Monte Carlo Simulation, for both: the Instructor and the Student Assignment Methods. It confirms the result from Subsection \ref{Sec Historical Assessment}, the Students Assignment Method ($ \sa $) yields better results than the Instructor Assignment Method ($ \ia $), and it shows that this result is not merely a particularity from our data set. Instead, it constitutes a robust one. 
\end{remark} 
\begin{table}[h!]
\def\arraystretch{1.4}
\small{
\begin{center}
\rowcolors{2}{gray!25}{white}
\begin{tabular}{p{1.5cm}|p{2.0cm}p{2.0cm}p{2.0cm}p{2.0cm}p{1.2cm}p{1.2cm}p{1.2cm}  }
    \hline
    \rowcolor{gray!80}
Experiment
&
\multicolumn{2}{c}{Random Normalization, $ \rho_{\mt} $ } &
\multicolumn{2}{c}{Expected Normalization, $ \gamma_{\mt} $ } 
& Enrollment 
& Sections
& Lecturers\\
\rowcolor{gray!80}
Number
& $ 100 \times \dfrac{v_{\ia} }{ \rho_{\ia} } $
& $ 100 \times \dfrac{v_{\sa} }{ \rho_{\sa} } $
& $ 100 \times \dfrac{v_{\ia} }{ \gamma_{\ia} } $
& $ 100 \times \dfrac{v_{\sa} }{ \gamma_{\sa} } $
& $ \upns $ 
& $ \sum\limits_{K\,\in\, \mathcal{I}} s_{K} $
& $ \upnt $
 \\[5pt]
\hline
1 &
0.3097 &	2.9059 &	0.3056 &	2.9044 &	1355	 &15	 & 6 \\
2 &
0.4595 &	3.2880 & 0.4588 & 3.2854	 & 1445 &	14 & 8 \\
3 &
0.4373 & 3.2655 & 0.4414 & 3.2653 & 1225 & 14 & 7 \\
4 &
0.4158 &	3.2689 &	0.4130 &	3.2663 &	1456	 &15	& 7 \\
5 &
0.4357 &	2.9680 &  0.4315 &	2.9651 &	1296	 & 14 & 6 \\
6 &
0.4943 &	3.1690 &	0.5053 &	3.1697 &	1547	 & 15 & 8 \\
7 &
0.5099 &	3.4486 &	0.5008 &	3.4439 &	1556 & 16	 & 8 \\
8 &
0.4937 & 3.3720  &  0.4841  &  3.3666 & 1532	& 16 & 8 \\
9 & 
0.4080 & 3.1254  &  0.4009  &  3.1301 & 1444 & 15	& 7 \\
10 &
0.4843 & 3.4498 &  0.4807  &  3.4454  & 1546 & 16  & 8 \\
\rowcolor{gray!80}
Mean &
0.4448 &	3.2261 &	0.4422 &	3.2242 & 1440.2 & 15.0 & 7.3
 \\
\hline
\end{tabular}
\end{center}
}
\caption{Monte Carlo Simulations Summary. The table shows a summary of the Monte Carlo Simulations with 10 experiments and 800 iterations each, for the \textbf{Differential Calculus} course.
}
\label{Tb Monte Carlo Experiments Summary}
\end{table}
\begin{algorithm}[H]
	\KwData{
		Database: \textit{Assembled\_Data.csv}
		Analyzed Course: DC, IC, ..., NM.\\
		Optimization Method: $ \mt \in \{\ia, \sa\} $. \\
		$ \upnt $ random variable distribution \\
		Numer of Iterations: $ NI $}
	\KwResult{Table of Relative Enhancement Values $ \rho_{\mt} $, $ \gamma_{\mt} $ for chosen method, course and academic performance variable.
	}
	\textbf{Initialization}\;
	\textbf{call} Algorithm \ref{Alg Random Setting MC Analytica Omega}\;
	$ \mathit{nt} \leftarrow $\textbf{compute} a realization of $ \upnt $\;
	\textbf{call} Algorithm \ref{Alg Instructors' performance}, \textbf{input}: (\textit{Assembled\_Data.csv},	\textit{Analyzed Course}, $ \apv $, Group Segmentation $ \big( I_{\ell}: \ell \in [L] \big) $)\;
	\For{ \text{iteration} $ \in [NI] $ }{
		 $ \mathit{list} \leftarrow $ \textbf{compute} a random list of $ nt $-lecturers\;
		\textbf{call} RandInputAlgorithm \ref{Alg Optimization Algoritm}, \textbf{input}:  
		(\underline{Group Assignment Matrix  \textit{G} , List of Lecturers \textit{L\_list} }, \textit{Analyzed Course}, $ \apv $, Group Segmentation $ \big( I_{\ell}: \ell \in [L] \big) $, $\mt$ )\;
		
		APV\_mt\_Assessment[\textit{iteration}]$ \leftarrow \big[ \rho_{\mt}, \gamma_{\mt} \big] $. 		
	}
	\caption{Monte Carlo Simulation Algorithm}
	\label{Alg Monte Carlo Analytica Omega}
\end{algorithm}
\begin{figure}[h!]
        \centering
        \begin{subfigure}[Example DC. Enhancement Results Pass Rate Monte Carlo Simulation. Instructor Assignment Method ($ \ia $). ]
{\includegraphics[scale = 0.38]{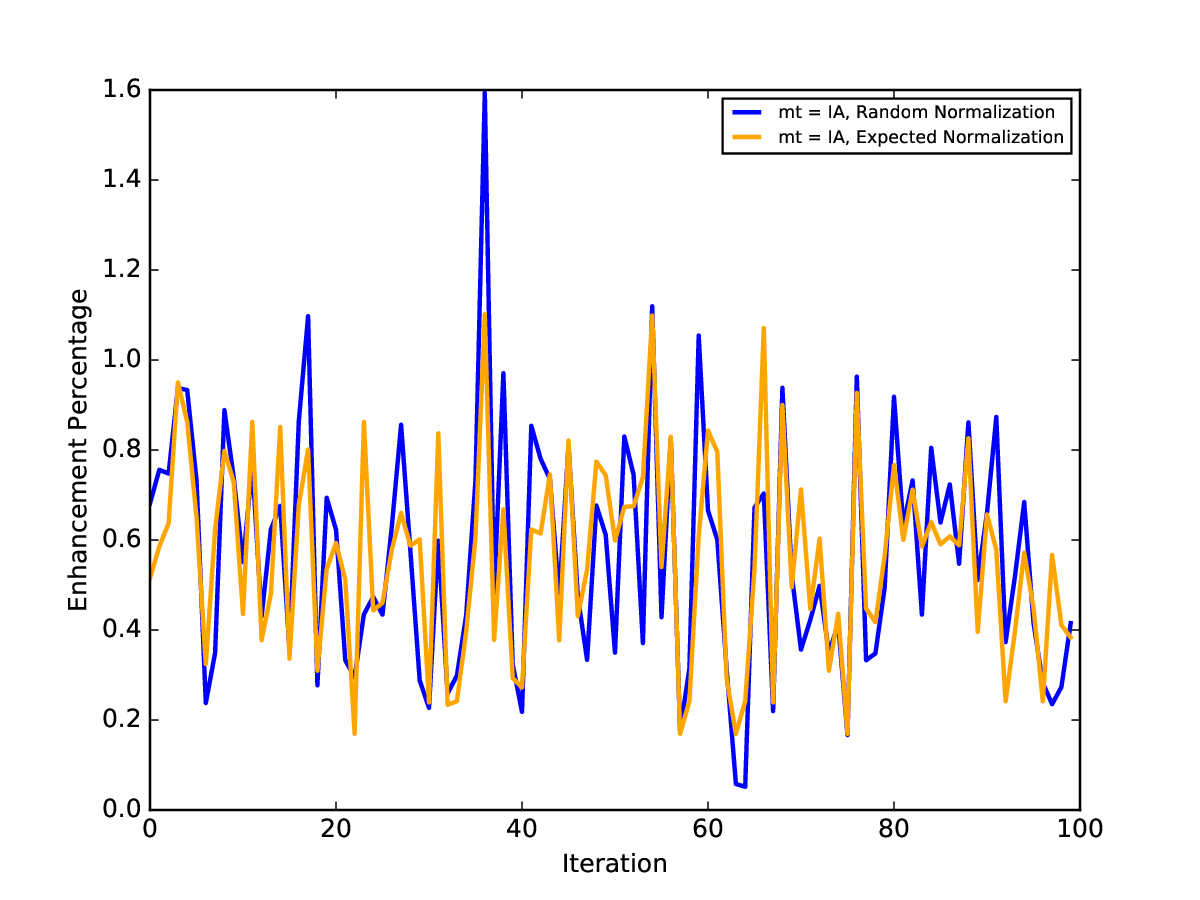} } 
        \end{subfigure}
        ~ 
          \begin{subfigure}[Example DC. Enhancement Results Pass Rate Monte Carlo Simulation. Student Assignment Method ($ \sa $). ]
{\includegraphics[scale = 0.38]{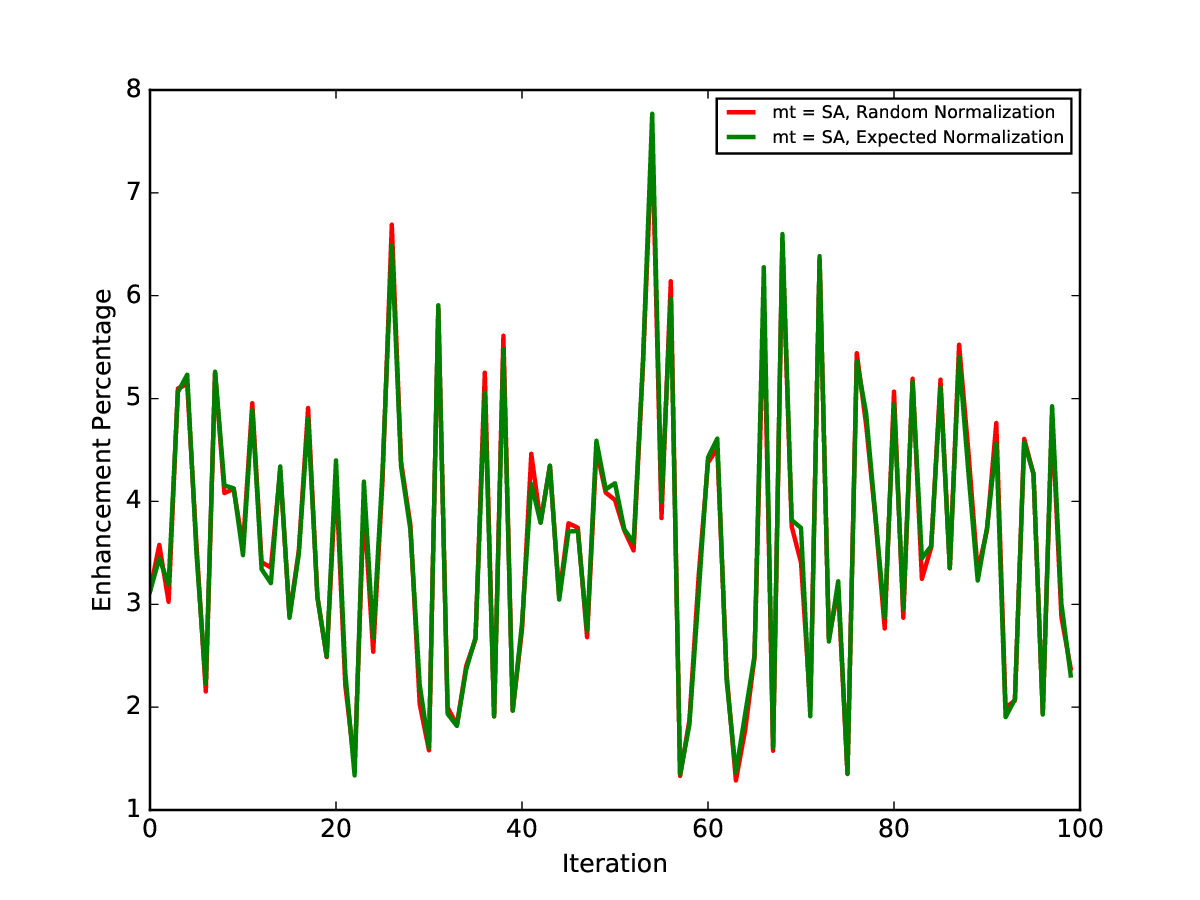} } 
        \end{subfigure}
        \begin{subfigure}[Example DC. Ces\`aro Means Enhancement Results Pass Rate Monte Carlo Simulation. Instructor Assignment Method ($ \ia $). ]
{\includegraphics[scale = 0.38]{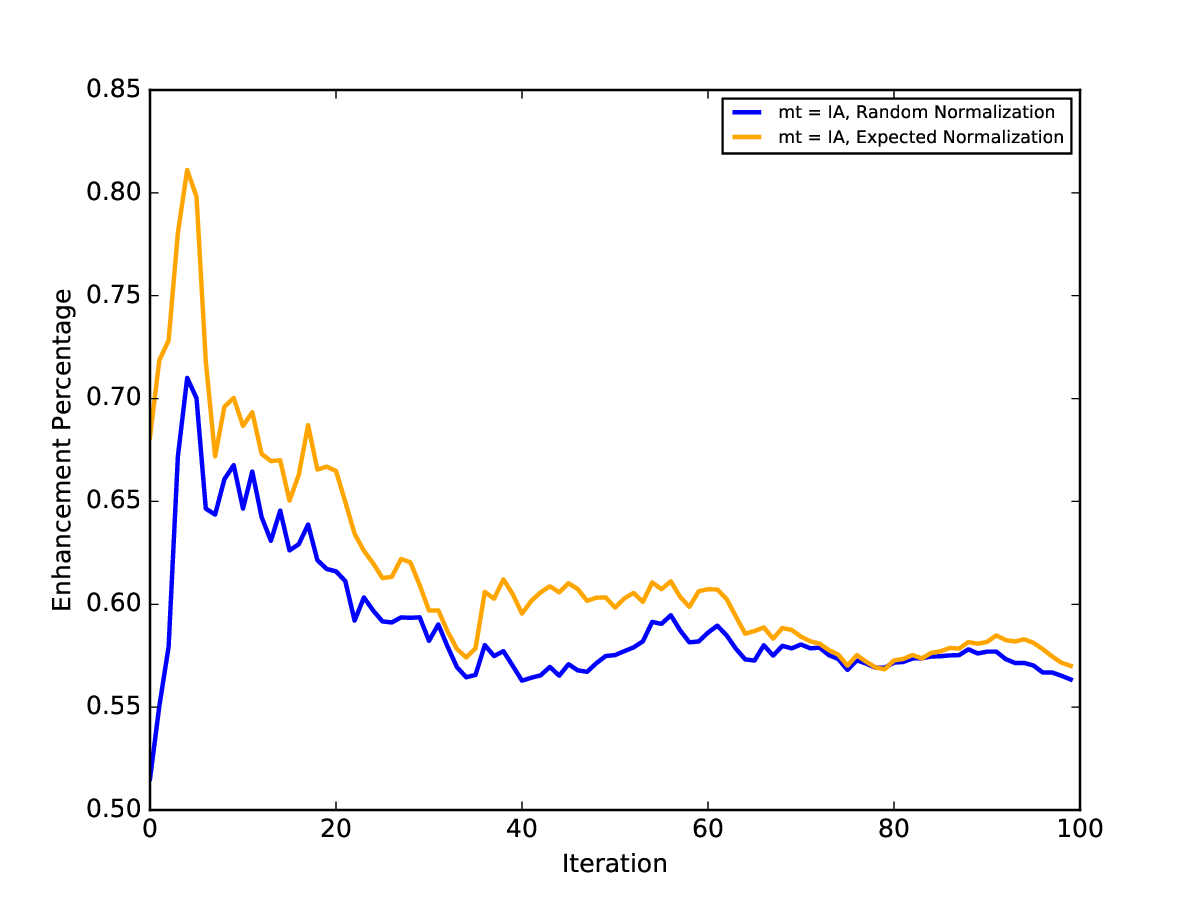} } 
        \end{subfigure}
        ~ 
          \begin{subfigure}[Example DC. Ces\`aro Means Enhancement Results Pass Rate Monte Carlo Simulation. Student Assignment Method ($ \sa $). ]
{\includegraphics[scale = 0.38]{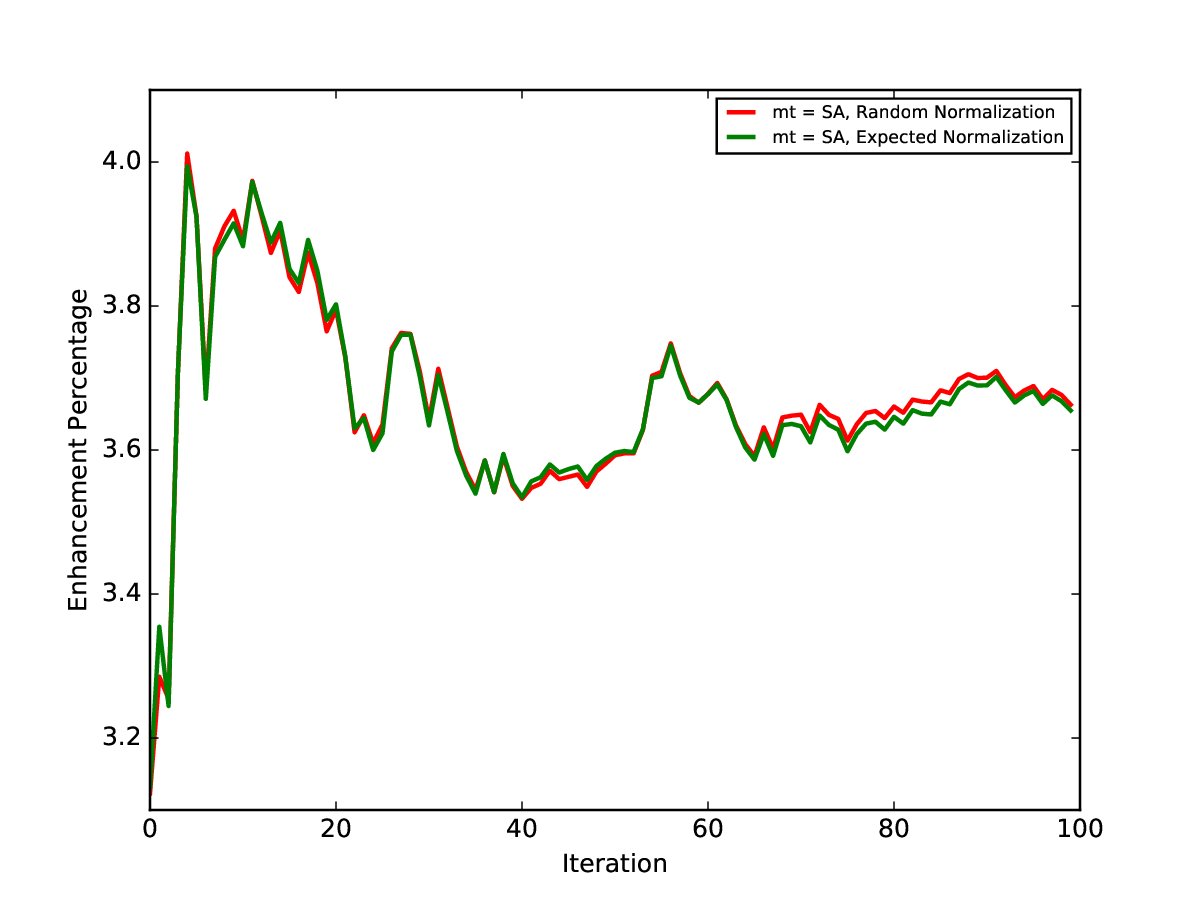} } 
        \end{subfigure}
\caption{Example: \textbf{Differential Calculus} course.
Enrollment of 1441 Students, 15 Sections, 8 Tenured Lecturers, 100 Iterations. All figures display the normalization $ \rho_{\mt} $ vs $ \gamma_{\mt} $ for $ \mt \in \{ \ia, \sa \} $. 
\label{Fig Asymptotic Enhancement} }
\end{figure}
\begin{table}[h!]
	\def\arraystretch{1.4}
	\scriptsize{
		\begin{center}
		\rowcolors{2}{gray!25}{white}
			\begin{tabular}{ c | c c c c c c c c c c | c }
				\hline
				\rowcolor{gray!80}
				\tiny{
				\diagbox 
				{SECTION}{SEGMENT} }
				&  \tiny{[0, 2.2]}
				&  \tiny{(2.2, 2.7]}
				&  \tiny{(2.7, 3.0]}
				&  \tiny{(3.0, 3.1]}
				&  \tiny{(3.1, 3.3]}
				&  \tiny{(3.3, 3.5]}
				&  \tiny{(3.5, 3.7]}
				&  \tiny{(3.7, 3.8]}
				&  \tiny{(3.8, 4.1]}
				&  \tiny{(4.1, 5.0]} 
				& Total \\
				\hline
				1&
				6 &	1 &	15 &	5 &	13 &	8 &	8 &	4 &	8 &	6 &	74 \\
				2 &
				10 & 6 & 14 &	3 &	5 &	7	& 11 &	5 & 8 &	5 & 74 \\
				3 &
				9 & 9	& 9	 & 8 &	6 &	3 &	9 &	1 &	15 &	5 &	74 \\
				4 &
				11 & 12	& 13 &	1 &	5 & 3 &	8 &	7 &	8 &	6 &	74 \\
				5 &
				5 &	9 &	12	& 7	& 11 &	6 &	9 &	2 &	5 &	9 & 75 \\
				6 &
				9 & 7 &	6 &	3 &	13 & 11 & 11 &	1 &	7 &	6 &	74 \\
				7 &
				12 & 5 & 7 & 2 & 7 & 10 & 12 &	6 &	21 & 7 & 89 \\
				8 &
				7 &	11 & 11 & 2 & 6 & 12 & 14 &	8 &	9 &	9 &	89 \\
				9 &
				11 & 18 & 14 &	7 & 6 &	12 & 10 & 7 & 10 &	9 &	104 \\
				10 &
				7 &	9 &	15 & 7 & 14 & 14 & 16 &	8 &	8 &	6 &	104 \\
				11 &
				14	& 8	& 20 &	4 &	15 & 17 & 14 &	0 & 13 & 14 & 119 \\
				12 &
				7 &	11 & 14 & 6	& 13 &	21 & 14 & 4	& 18 & 11 &	119 \\
				13 &
				15	& 15 &	15 & 9	& 15 &	16 & 11 & 3 & 11 &	9 &	119 \\
				14 &
				14	& 14 &	20 & 5	& 17 & 15 &	12 & 4 & 9 & 9 & 119 \\
				16 &
				15 & 13 & 25 &	6 &	15 & 16	& 10 &	9	& 17 &	8 & 134 \\
				\hline
				\rowcolor{gray!80}
				Total &
				152	& 148 & 210	& 75 & 161 & 171 & 169 & 69	& 167 &	119 & 1441
				\\
				\hline
			\end{tabular}
		\end{center}
	}
\caption{An example of a random realization of Algorithm \ref{Alg Random Setting MC Analytica Omega}, i.e. a group matrix assignment $ G $ and a number of tenured lecturers $ \upnt $ for the \textbf{Differential Calculus} course.
}\label{Tb Sections Realization Example}
\end{table}
\begin{table}[h!]
\def\arraystretch{1.4}
\small{
\begin{center}
\rowcolors{2}{gray!25}{white}
\begin{tabular}{p{1.5cm}|p{2.0cm}p{2.0cm}p{2.0cm}p{2.0cm}p{1.2cm}p{1.2cm}p{1.2cm}  }
    \hline
    \rowcolor{gray!80}
Course
&
\multicolumn{2}{c}{Random Normalization, $ \rho_{\mt} $ } &
\multicolumn{2}{c}{Expected Normalization, $ \gamma_{\mt} $ } 
& Enrollment 
& Sections
& Lecturers\\
\rowcolor{gray!80}

& $ 100 \times \dfrac{v_{\ia} }{ \rho_{\ia} } $
& $ 100 \times \dfrac{v_{\sa} }{ \rho_{\sa} } $
& $ 100 \times \dfrac{v_{\ia} }{ \gamma_{\ia} } $
& $ 100 \times \dfrac{v_{\sa} }{ \gamma_{\sa} } $
& $ \upns $ 
& $ \sum\limits_{K\,\in\, \mathcal{I}} s_{K} $
& $ \upnt $
 \\[5pt]
\hline
DC &
0.4448 &	3.2261 &	0.4422 &	3.2242 & 1440.2 & 15.0 & 7.3 \\
IC	& 0.3267	& 2.8094	& 0.3196	& 2.8104	& 1068.0	& 8.2	& 5.1 \\
VC	& 0.1684	& 1.9797	& 0.1800	& 1.9923	& 586.6	& 4.1	& 2.4 \\
VAG	& 0.4070	& 3.1366	& 0.4079	& 3.1474	& 1080.8	& 14.5	& 6.2 \\
LA	& 0.2131	& 3.2906	& 0.2009	& 3.2788	& 1078.2	& 8.0	& 4.2 \\
ODE	& 0.4323	& 5.6270	& 0.4269	& 5.6332	& 798.2	& 6.4	& 3.1 \\
BM	& 0.5706	& 3.0775	& 0.5909	& 3.0791	& 910.9	& 11.1	& 3.0 \\
NM	& 0.3750	& 3.3825	& 0.3405	& 3.3920	& 263.1	& 2.3	& 1.7 \\
\hline
\end{tabular}
\end{center}
}
\caption{A summary of the Monte Carlo Simulations with 800 iterations for each course.
}\label{Tb Monte Carlo Experiments All Courses Summary}
\end{table}
%
%
%
%
%
%
%
%
%
\section{Conclusions and Future Work}\label{Sec Conclusions}
%
%
%
%
The present work delivers several conclusions. 
\begin{enumerate}[I.]
\item From the modeling point of view
\begin{enumerate}[(i)]
\item A method has been implemented aimed to increase the academic performance for massive university lower division courses in mathematics. It is based on integer programming and big data analysis to compute the associated cost functions, while the constraints (such as the number of sections and corresponding capacities) are defined by administrative sources. The integer programs come from two mechanisms: assign instructors optimally ($ \ia$ method) or assign students optimally ($ \sa $ method).   

\item The academic performance was explored using two measures; Pass Rate and Grade. After correlation analysis of the data, it is determined that the one relevant factor, known at the time when the semester begins and incident on these statistical variables is the $ \gpa $. Consequently the profiling of students as well as the student body composition is defined in terms of the $ \gpa $ (see \textsc{Table} \ref{Tb Historical Enhnacements})

\item The historical assessment of the method yields poor enhancement levels for the Grade variable, due to its typical statistical robustness. However, the Pass Rate yields more satisfactory results; good enough to pursue a deeper analysis such as the method's randomization and its asymptotic assessment, presented in \textsc{Section} \ref{Sec Randomization and Asymptotic Assesment}.

\item The asymptotic analysis of the algorithm is done by randomizing the enrollment population and the administrative factors, statistically based on the empirical observations reported in the database \textit{Assembled\_Data.csv}. The Monte Carlo experiments establish that the method does not deliver a fixed value of relative enhancement, it depends on the starting parameters $ \big(\upns, G, \upnt\big) $ whose remarkable randomness inherit uncertainty to the algorithm's output values. 

\item Computing a weighted average across the courses by crossing the tables \ref{Tb Enrollement Random Variable} and \ref{Tb Monte Carlo Experiments All Courses Summary}, gives a rough estimate of 3.3 percent  full scale benefit, if the students assignment method ($ \sa $) is implemented. This is approximately 240 extra students per semester passing their respective courses which, in the long run represent a significant gain for the Institution.

\item The algorithms \ref{Alg Analytica Omega} and \ref{Alg Monte Carlo Analytica Omega}, could have been adjusted to keep only the sections with tenured lecturers. However, the Authors chose to discard this artificial setting because it is biased with respect to the study case.

\item It is the perception of the Authors that no general conclusions can be derived for the method's enhancement level. On one hand it is sufficiently general and flexible to be implemented at any Institution with massive courses and therefore big databases available. On the other hand, the experiments performed in the present work, suggest that its effectiveness needs to be evaluated on a case-wise basis.

\item Considering age as a factor is also possible by merely applying the segmentation process described in Section \ref{Sec Segmentation Process} (\textsc{Algorithm} \ref{Alg Segmentation of Students} with an adequate number of segmentation intervals $ \big( \widetilde{I}_{\ell}: \ell \in [\widetilde{L}]\big) $). First, computing the lecturers performance conditioned to the \textit{Age} variable as in \textsc{Section} \ref{Sec Computation Lecturer Performance} (\textsc{Algorithm} \ref{Alg Instructors' performance}, output $ T_{\textit{Age}} $). Second, weighting its impact according to the correlation values, namely the costs table in \textsc{Equation} \eqref{Eq Costs Table APV} can be modified as
\begin{equation}\label{Eq Costs Table Age and APV}
C = \frac{4}{5} \,  T_{\textit{APV} } G + \frac{1}{5}T_{\textit{Age}} \widetilde{G}.
\end{equation}
Here it is understood that the group matrix $ \widetilde{G} $ is constructed according to the \textit{Age} variable segmentation $ \big( \widetilde{I}_{\ell}: \ell \in [\widetilde{L}]\big) $. The weighting coefficients were proposed, according to the correlation with the \textit{Grade} variable reported in \textsc{Table} \ref{Tb Correlations DC}: \textit{Age}: 0.2, $ \gpa $: 0.8, i.e., the second is 4 times the first one (see  \cite{DeGiorgiPellizaariWoolstonGui} for further discussion on these type of models). Yet again, the flexibility of the method, allows to introduce in the same fashion any number of variables fitting to the case at hand.
\end{enumerate}
\item From the economy point of view
\begin{enumerate}[(i)]
\item A 3.3 \% enhancement for the method's benefit may seem low at first sight. However, it is important to stress that this enhancement corresponds to a detailed treatment of the tenured lecturers only, while the adjunct lecturers are treated in general terms because of insufficient data as they are unstable personnel. Tenured lecturers represent a fraction of less than 50 percent from the involved faculty team as \textsc{Table} \ref{Tb Monte Carlo Experiments All Courses Summary} shows. Consequently, should the stable personnel fraction increase, the method would deliver more accurate and perhaps more optimistic results. 

\item The method presented in this work offers a mechanism for higher education institutions to help their students improving their pass rates and grades. This is done by solving two different social welfare schedules ($ \ia $ and $ \sa $ methods). Under this approach, the University is considered as an agent that provides education, and as a \textbf{rational regulator agent}, capable to optimally allocate some of its resources for enhancement of social welfare of its students body.

\item The method has two important features that are particularly relevant in countries like Colombia where the drop out rates from college are high and the investment in higher education is low: $ \bullet $ By helping students to improve their pass rates and grades, it is alleviating the problem of high drop out rates; $ \bullet $ the method does not require major money investment from the University in order to be implemented. In theory, only the data and a capable person are required to implement it. 

\item The work also provides a way to measure (and monitor) how far from the Pareto equilibrium is an Institution at a given time. This important because it provides a way to determine whether the expected enhancement results are being achieved or not. 
\end{enumerate}
\item From the future work point of view
\begin{enumerate}[(i)]
\item This paper has worked two methods, assign instructors while keeping the students fixed ($ \ia $) and assign students while keeping the instructors fixed ($ \sa $), both of them come down to a linear optimization problem,  \ref{Pblm Instructors Assignment} and \ref{Pblm Students Assignment Problem} respectively. However, moving both instructors and students simultaneously is no longer a linear, but a bilinear optimization question (see \cite{Orlov, CapraraMonaci}). This view will be further explored in future work.  

\item So far, the present work assumed that allocating students and/or instructors is a decision centralized by the administrative departments of the analyzed Institution. However, in our study case, student location is decided differently, using a $ \gpa $ competition-based mechanism to assign priority starting from the highest to the lowest scorers. This competitive scenario is better modeled using game theory which will be explored in future work. 

\item The algorithm presented in this work offers a mechanism for higher education institutions to help their students improving their pass rates and grades. This is done by solving two different social welfare schedules ($ \ia $ and $ \sa $ methods). The work also provides a way to measure (and monitor) how far from the Pareto equilibrium is an Institution at a given time. As mentioned in the economic justification (Subsection \ref{economic_justification}), this is particularly relevant in countries like Colombia where the drop out rates from college are high.
\end{enumerate}
\end{enumerate}
%
%
%
%
%
%
%
\section{Acknowledgements}
%
%
%
%
The Authors wish to thank Universidad Nacional de Colombia, Sede Medell\'in for its support in the production of this work, in particular, to the Academic Director of the University for allowing access to their databases for this study. The first author was supported by grant Hermes 45713 from Universidad Nacional de Colombia, Sede Medell\'in. 
%
%
%
%
%
\bibliographystyle{plain}
%

%
%
%
%
%
%
%
\end{document}